\documentclass[]{article}
\usepackage{etex}
\usepackage{amssymb}
\usepackage{amsmath}
\usepackage{amsthm}
\usepackage[all]{xy}
\usepackage{hyperref}
\usepackage{bbm}
\usepackage{bm}
\usepackage{float}
\usepackage[font={footnotesize}]{caption}
\usepackage[round]{natbib}
\usepackage{mathtools}
\usepackage{setspace}
\usepackage{comment}
\usepackage{tikz}
\usepackage{pgf}
\usetikzlibrary{arrows,automata}
\usetikzlibrary{positioning}
\usetikzlibrary{shadows,shapes.arrows,decorations}
\usepackage{pgfplots}
\pgfplotsset{compat=newest}
\usepackage[all]{xy}
\usepackage[]{fullpage}
\usepackage{todonotes}
\providecommand{\keywords}[1]{\textit{Keywords:} #1}
\usepackage{tabularx}  
\usepackage{ragged2e}  
\newcolumntype{Y}{>{\RaggedRight\arraybackslash}X} 
\usepackage{booktabs}
\usepackage{authblk}

\newtheorem{theorem}{Theorem}

\newtheorem{definition}{Definition}
\newtheorem{example}{Example}

\newtheorem{lemma}{Lemma}

\newtheorem{property}{Property}
\numberwithin{definition}{section}
\numberwithin{theorem}{section}
\numberwithin{example}{section}
\numberwithin{lemma}{section}
\numberwithin{corollary}{section}
\numberwithin{proposition}{section}
\numberwithin{equation}{section}

\newcommand{\dr}{\textnormal{d}}
\newcommand\independent{\protect\mathpalette{\protect\independenT}{\perp}} \def\independenT#1#2{\mathrel{\rlap{$#1#2$}\mkern2mu{#1#2}}}

\newcommand{\T}{\textnormal{T}}
\newcommand{\E}{\mathbb{E}}
\newcommand{\V}{\mathbb{V}}
\newcommand{\Gr}{\mathcal{G}}

\usepackage{hyperref}

\begin{document}
\title{Coherent Frameworks for Statistical Inference serving Integrating Decision Support Systems}
\author[$\dag$]{Jim Q. Smith\footnote{\textit{Address for correspondence}: Jim Q. Smith, Department of Statistics, The University of Warwick, CV4 7AL Coventry, UK.
E-mail: J.Q.Smith@warwick.ac.uk}}
\author[$\dag$]{Martine J. Barons}
\author[$\dag$ $\dag$]{Manuele Leonelli}
\affil[$\dag$]{Department of Statistics, The University of Warwick}
\affil[$\dag$ $\dag$]{Instituto de Matem\'{a}tica, Universidade Federal do Rio de Janeiro}
\date{\vspace{-5ex}}

\maketitle
\begin{abstract}
A subjective expected utility policy making centre managing complex, dynamic systems needs to draw on the expertise of a variety of disparate panels of experts and integrate this information coherently. To achieve this, diverse supporting probabilistic models need to be networked together, the output of one model providing the input to the next. In this paper we provide a technology for designing an integrating decision support system and to enable the centre to explore and compare the efficacy of different candidate policies. We develop a formal statistical methodology to underpin this tool. In particular, we derive sufficient conditions that ensure inference remains coherent before and after relevant evidence is accommodated into the system. The methodology is illustrated throughout using examples drawn from two  decision support systems: one designed for nuclear emergency crisis management and the other to support policy makers in addressing the complex challenges of food poverty in the UK.

\vspace{0.35cm}
\noindent\keywords{Bayesian multi-agent models, causality, coherence, decision support, graphical models, likelihood separation.}
\end{abstract}

\section{Introduction}
Using a probability model for decision support for a single user has many advantages: as well as ensuring coherence, and hence transparency, recent computational advances have enabled such support to be fast, even when large amounts of structured information needs to be accommodated. However, the 21st century has seen the advent of massive models which need to be networked together to provide appropriate decision support in increasingly complex scenarios (see e.g. Figure \ref{networkino}). Each component of such a network is itself often informed by huge data sets. In these contexts,  users are typically decision \emph{centres} where both  users and experts are teams rather than individuals. Such centres often need a tool that can draw together inferences in this plural environment and integrate together expert judgements coming from a number of different panels of experts where each panel is supported by their own, sometimes very complex, models.  

\begin{figure*}
\begin{center}
\begin{tikzpicture}[->,>=stealth',shorten >=1pt,auto,node distance=1.5cm,
                    semithick]
 \tikzstyle{every state}=[fill=white,draw,text=black, shape=rectangle, rounded corners,every text node part/.style={align=center}]
  \node[state]  (A)      [fill=black!20, draw=black!100, text=black]              {Power \\ plant};
  \node[state]         (B) [right = of A, draw=green!100, fill=green!20] {Water \\ dispersal};
   \node[state] (C) [draw=blue!100,  fill=blue!20,  right = of B]{Human \\ absorption};
\node[state] (D) [below of=A, fill=black!20, draw=black!100, text=black ]{Source \\  term};
\node[state](E)[below of=D, draw=green!100, fill=green!20]{Air \\ dispersal};
\node[state](F)[below of=B, draw=green!100, fill=green!20]{Deposition};
\node[state](H)[below of=F, draw=brown!100, fill=brown!20]{Political \\ effects};
\node[state] (L)[draw=blue!100,  fill=blue!20, text=black, below of =  C]{Animal \\ absorption};
\node[state](I)[below of=L, fill=yellow!20, draw=yellow!100]{Costs};
\node[state](G)[draw=red!100, fill=red!20, right=1cm of L]{Human \\  health};
\path (A) edge (D)
        (D) edge (B)
             edge (E)
         (B) edge (C)
              edge (L)
             edge (F)
          (E) edge (H)
               edge (F)
           (L) edge (C)
         (C) edge (G)
         (F) edge (L)
             edge (I)
             edge (C)
          (L) edge (I)
         (G) edge (I)
          (I) edge (H);
\end{tikzpicture}
\end{center}
\caption{A plausible network of models for a decision support system in nuclear emergency management (from \cite{Leonelli2013d}), where nodes equally colored are associated to the same domain of expertise. \label{networkino}}
\end{figure*}
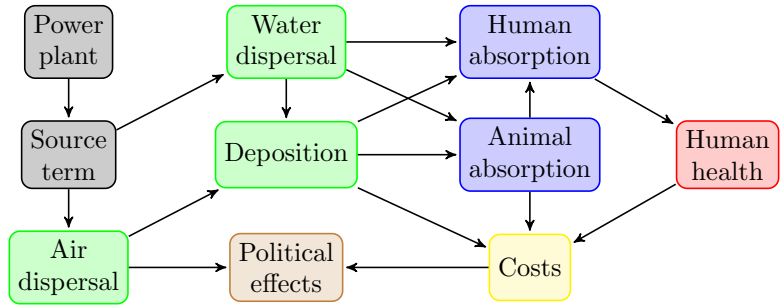 

Although, increasingly, many of these component expert panels are supported by \emph{probabilistic} models, it is natural but usually inappropriate to commission a single comprehensive probabilistic model over the whole composite, except in the case of relatively small systems. Such an overarching probability model would be huge and, perhaps more critically, unless there existed shared structural assumptions, no single centre could realistically `own' all the statements about the full joint distribution of the hundreds of thousands of diverse random variables in the aggregate system. Furthermore, and from a more practical perspective, even if it \emph{were} possible to build such a system, typically, in the types of domain we address in this paper, the different component systems are being constantly revised by the relevant panels to accommodate new understanding, science and data.  Any overarching probabilistic model would therefore  quickly become obsolete: the judgements it embodies would no longer reflect current understanding.

In this paper we argue that what is often needed instead of a single, overarching probability model is an \emph{integrating decision support systems} (IDSS). This would process only carefully selected probabilistic outputs from each contributing component - those statistics on which expected utilities would need to depend. It would then combine these dynamically evolving expert judgements together appropriately to provide the basis for a coherent assessment. In particular, benchmark numerical efficacy scores for each candidate policy would then be calculated from the individual components' outputs and thus help determine the efficacy of each of the different options considered by the decision centre. 

In fact, theory and methods for breaking up probability models into autonomous subcomponents are already well developed by agent-based modellers \citep[see e.g.][]{Xiang2002}, albeit when addressing automated decision-making rather than the anthropomorphic decision-making we address here.  Perhaps even more relevant to this development is the seminal paper by \citet{Mahoney1996} who apply general engineering principles to develop protocols for the coherent integration of evidence over a diverse panel of experts informing a problem. Although they quickly focus their ideas down on to the specific Bayesian network (BN) model class, their framework is nevertheless a valuable one and is currently being exploited within the context of object-oriented BNs (OOBNs) by a number of authors \citep[e.g.][]{Johnson2014, Johnson2012}. 

For the purposes of this paper we return to the more general setting described initially by \citet{Mahoney1996}. Here we apply analogous principles to even larger problems that they envisaged, developing a sound statistical methodology that justifies the use of an IDSS. The challenges faced by chains of panels of experts using data and models to deliver probabilistic beliefs has been noted in \citet{French2011}, who argued that very little had been said on the issue so far. Some early work on this type of problem was suggested for nuclear emergency management in \citet{French1991}, however, as for as we are aware, this paper is the first to approach these types of problem from a methodological and statistical viewpoint.  

To highlight the relevance of an IDSS to support policymakers in current, complex domains we start by discussing two applications where we have observed the necessity of knitting together different models into a coherent whole,  which together provoked the methodological development within this paper.

\subsection{Two systems we have appraised}
\label{appl}
\subsubsection{RODOS and nuclear emergency management}
In 1986 an explosion at one of the  reactors of the Chernobyl nuclear power plant released a radioactive plume into the environment contaminating large areas of the former Soviet Union. To protect people and food stocks, measures were taken by the governments of the affected countries. Further different and often conflicting responses were taken by many European countries after the accident, confusing the  public, and leading to an ineffective implementation of countermeasures \citep{Papamichail2013, Walle2008}. It was therefore quickly recognised that a comprehensive response to nuclear emergencies within the European community was needed. To achieve this, a common decision support system (DSS) for off-site emergency management was commissioned. Several institutes in Europe then started the development of the Real-time On-line DecisiOn Support system (RODOS) for nuclear emergencies, including uncertainty handling methodologies which would provide consistent predictions unperturbed by national boundaries \citep{Ehrhardt1993}. One author of this paper was heavily involved in this development.

Alongside the development of RODOS, in the early 1990s the International Chernobyl Project began to explore the factors that drove decision making about which protective measures to adopt after the Chernobyl accident. Since, at the time, many different parties and institutions were involved in this decision making process, the study was organised through conducting five decision conferences, where simple multicriteria decision analysis  models were used to explore the preferences and the beliefs of the different parties \citep{French2009}. The analyses performed during these meetings clearly showed that factors implemented in cost-benefit analyses, usually performed in nuclear emergency management, could not fully describe the preferential structure of the group. It was therefore decided that multicriteria methods had to be included in any operational DSS like RODOS designed for nuclear emergency response. Such a DSS would then combine scientific knowledge about the likelihood of different events with the value judgements about these to rank different agreed available policies and both facilitate the exploration and create a deeper understanding of the problem at hand. A sample output from RODOS is shown in Table \ref{table:policy} where high scoring countermeasures are detailed together with a breakdown of the impact of the policies on the relevant factors identified through decision conferencing. This type of supporting capabilities, presented in this and other more refined forms \citep{Papamichail2003,Papamichail2005}, were found to be vital for integrated decision support. This is because empirical research has shown that decision makers do not accept the suggestions of a system which does not provide a rationale for the outputs it produces, even if these outputs happen to be accurate \citep{Papamichail2013}.

\begin{table}
\begin{tabularx}{\textwidth}{@{} Y Y Y Y Y Y @{}} 
\toprule
\texttt{Strategy}&\texttt{Number relocated (thousands)}& \texttt{Number protected by other means (thousands)} & \texttt{Estimated number of fatal cancers averted}  & \texttt{Estimated number of hereditary effects averted}& \texttt{Cost (billions of roubles)}\\
\midrule
\texttt{SL2\_2}&\texttt{706}&\texttt{0}&\texttt{3200}&\texttt{500}&\texttt{28}\\
\texttt{SL2\_10}&\texttt{160}&\texttt{546}&\texttt{1700}&\texttt{260}&\texttt{17}\\
\texttt{SL2\_20}&\texttt{20}&\texttt{686}&\texttt{650}&\texttt{100}&\texttt{15}\\
\texttt{SL2\_40}&\texttt{3}&\texttt{703}&\texttt{380}&\texttt{60}&\texttt{14}\\
\bottomrule
\end{tabularx}
\caption{A sample RODOS output breaking down the score of various strategies for the different factors deemed relevant. 
\label{table:policy}}
\end{table}

An evaluation of the potential unfolding trajectories of an emergency was achieved by pasting together the outputs of  a suite of different subsystems (or modules). Each  such module provided estimates and forecasts for a different aspect of the emergency.  At that time there was an acute awareness that uncertainty evaluations had a critical role to play in such systems \citep{Smith1997}. However it was also the case that the formal accommodation of such uncertainties could not be made homogeneously. The system needed to use a variety of deterministic and stochastic methodologies to guide the estimation and the forecasting of the various quantities relevant to the domain under study \citep[][]{Ehrhardt1997,Smith1997}.  A few modules were statistical in nature, but others were guided by fuzzy logic and many others were entirely deterministic.

Once uncertainty management for these networks of systems became acknowledged as central to the effective implementation of the composite system, fully probabilistic component modules  began to be developed for communicating both the relevant panel's forecasts and their associated uncertainties. There were several examples of these. For instance, a source term module estimating the likelihood of a release of contamination from the plant was built \citep{French1995}. Others included atmospheric diffusion and deposition models describing the spread of contamination \citep{Smith1999,De2011}. Additional subsystems modelled the effect that the spread might have because of the exposure of  humans, animals and plants \citep{Richter2002,Zheng2009}. However these developments were patchy. This and the significant extra computational costs meant that any nuanced full integration of uncertainty evaluations associated with the whole system was severely inhibited. 

It was not only the lack of technology and the heterogeneity of uncertainty outputs of the different component models that challenged a comprehensive and proper accommodation of uncertainty. It was recognised early on that, \emph{even} if all models could produce faithful and consistent representations of the outputs' uncertainty for each module, it was  not at all clear at that time what formulae to use to combine these judgements and what background information justified that use of these forecasts. The fact that the component modules were designed to work independently of one another and that the exact judgements encoded within them were often contributed by very differently informed groups of scientists made this formal combination especially challenging. Because of these computational and methodological constraints, the modules' outputs ended up being collated together in a simple, essentially deterministic way by transferring from one module to another a single vector of means about what might happen and hence effectively ignoring any uncertainty associated with them. As statisticians, we appreciate how  such a na{\"i}ve method can be very misleading  \citep[see e.g.][]{Leonelli2013d,Leonelli2015}. 

Since that time the need for integrated support tools addressing other threats from uncertain environments has been recognised. Perhaps one of the most critical of these concerns food security.

\subsubsection{UK food security}
Food security, once thought to be a problem confined to low-income countries, is increasingly being recognised as a matter of concern in the UK \citep{DEFRA2009, ElliottReview2014,Lambie-Mumford2014,  Loopstra2015}, USA \citep{HFSS2012}, Canada \citep{Loopstra2012} and other wealthy nations. At a country level, all these nations appear to be the most food secure in the world \citep{GlobalFoodSecurityIndex}. But at household level, the story is quite different. Enabling its citizens to have access at all times to sufficient nutritious food for an active and healthy life is a key responsibility of governments, but achieving this is not straightforward. At first glance, UK household food security may seem to be a simple case of demand and supply. However on closer inspection the system is shown to be highly complex, especially from the point of view of policymakers, who endeavour to intervene on the system to produce specific responses \citep{DowlerMorris2000, Drewnowski2004, Dowler2015}. 

The food system is global, multifaceted and influenced by a huge number of public and private actions and uncontrolled factors such as weather and climate. This leads to a great deal of uncertainty about how any policy decision or strategy will play out. We are now taking up this new challenge. Since the domain is still developing we have had the opportunity to develop an overarching methodology to manage this system through the integration of diverse \emph{probabilistic} systems and through this a proper management of uncertainty. We are currently working with Warwickshire County Council to develop an IDSS to support decision-making around household-level food poverty. During this process we have recognized that a DSS to support policymakers in this new domain of application would need to have many features in common with RODOS, whilst embedding a complete uncertainty handling both within and between the constituent modules each informed by different panels of experts.

 Firstly, the system is multifaceted and heterogeneous and requires as inputs the judgements from different panels of experts in diverse disciplines including insight about factors elevating the risk to food security of households (from sociologists and local authorities), judgements about the effects of malnutrition on the population (from doctors and nutritionists), estimates of the availability of food in supermarkets and other outlets (delivered by supply chain experts) and forecasts of the yield of crops in a particular season (by crop experts and official statistics). Unless properly structured, this expert information is liable to conflict where two or more panels can sometimes deliver contradicting expert judgements about a shared random variable. If the system admits such contradictions then this can obviously threaten the coherence of the system as a whole so that its outputs become compromised. For instance, both estimates of cost of oil and weather forecasts affect food production, food transport and the ability of households to access food. If these latter variables are under the jurisdiction of different panels, any integrating system should surely embed common estimates of distributions over the cost of oil and weather forecasts and not contradicting ones. Otherwise how could it ever be coherent and justifiable? 
 
 Secondly, as was also true for the RODOS system described above, information that is available to inform each element of the system is patchy. Some parts of the system, for instance food production and household demography, are well modelled and informed by many different datasets and a great deal is known about the impacts of malnutrition \citep{Elia2010}. On the other hand, other subsystems are not necessarily so well informed. For instance, there is considerable uncertainty about world food availability because food imports, exports and prices are  highly variable and affected by a large number of factors such as weather, wars, international relations and even, unexpectedly, by another country's internal financial regulations \citep{Lagi2012, Lagi2012a}. Lastly, it is clear that any IDSS supporting this domain needs to be dynamic, not only because the food system is highly seasonal but also because there are many time steps for the consequences of a chosen policy  as well as shocks experienced by the system to unfold.

There is now a wealth of qualitative information in the sociological literature of the causes and impacts of food security at household level in the UK \citep{Dowler2001, Holmes2008, APPG2014,Dowler2014APPGresponse}. This provides a sound basis for modelling the qualitative structure of the system. Within the UK, \cite{Dowler2001} describes the strategies used by families to negotiate poor access to food resources. It has now been realised that strategies more subtle than taxing nutrient-poor food and subsidising nutrient rich foods are required to effect change in purchasing habits \citep{Darmon2014, Darmon2014a, Dowler2010, Elia2010, Friel2004, Holmes2008, PovertyPremium2014}.  Whilst the UK does not measure household-level food security, despite calls to do so for the last 20 years, culturally similar countries, the USA and Canada, regularly deploy an 18-question survey called HFSS \citep{HFSS2012}. This survey has also been implemented in small-scale studies in the UK \citep{Holmes2008, Pilgrim2012} providing insight into the extent to which USA and Canada findings may be used to provide a first approximation for the food poverty crisis now unfolding within the UK.

Many factors affecting household level risks for food security (disposable income, socio-economic status, social security levels, employment, housing and energy costs, access to credit, and so on) can be found for the UK in official statistics, at national, regional, county and sub-county level (lower layer super output areas (LSOA) and middle layer super output areas (MSOA)). For the Warwickshire decision-makers, MSOA level is the most appropriate, and a repository of data has been collated for use in a Warwickshire IDSS. 
 But where data are not available at this level, they can be modelled using the spatial granularity that is available or estimated from it by structured elicitation of expert opinion.  On the food supply side, there are some localities which are food deserts: so called because there local population is predominantly in low-income households and  the opportunity for profit is insufficient to incentivise supermarkets to site full-service stores. In such localities, households must either purchase food locally, typically at small convenience store  or finance transport to reach the larger stores. The small local stores often stock a smaller range of goods and often a much smaller range of fruits and vegetables often at a significantly higher price \citep{Dowler1999}.  The number of regions in the UK which are food deserts is likely to increase as UK supermarket giants close stores to address their falling profits \citep{GoldmanSachs2014}.  
 
Whilst  numerous DSSs exist to model aspects of the systems such as for supermarket siting and to inform pesticide and fertilizer use \citep{Decuyper2014,Efendigil2009, Hernandez2000, Kuo2002}, the complex problem of developing a shared methodology that can guide the accommodation  of the wide range of   expertise and provide the information required to evaluate the efficacy of various policies designed to address food poverty issues has yet to be attempted. This and other applications, like the needs of the RODOS project, have motivated the methodological developments we present below.

\subsection{Some examples of established probabilistic composites}
Although, for example, in food security we are only now in the process of building a suite of fully probabilistic modules to populate the IDSS and in nuclear emergency management some parts of the system were still deterministic, for the purposes of this paper we henceforth assume that we are in the position where there already exist probabilistic models that will describe the different components that our integrating system will network together. The type of probabilistic DSSs which will form the components of our integrated system are now widely available in a variety of forms. Two of the most common probabilistic models for multivariate systems designed to be used by a single agent or panel in non-dynamic environments are BNs and influence diagrams. Both of these frameworks, whilst still being refined, are now well developed and have been applied to ever larger systems \citep{Aguilera2011,Gomez2004,Cowell2011,Molina2010}. 

The component domains now modelled are often complex, dynamic and can themselves collate diverse inputs. Most recently, significant methodological developments using, for example, object-oriented code have enabled these models to become progressively more expressive, efficient and applicable as components within the types of dynamic environments we address here \citep{Koller1997,Koller2001,Murphy2002,Nicholson2011}. However, BNs and their dynamic analogues are not the only framework around which probabilistic models have been built. Hierarchical Bayesian models of large-scale temporal spatial processes modelling, for example, the development of epidemics of one sort or another  are well established \citep{Best2005,Jewell2009,McKinley2014}. Other modelling tools that have also recently appeared are those based on probability trees \citep{Thwaites2010,Smith2010}.  These have been successfully applied in a number of applications where the potential development of a scenario is heterogeneous \citep{Smith2008,Smith2010}.

Another class of probabilistic models used in large, complex systems are based on probabilistic emulators \citep{Kennedy2001,OHagan2006}. These have, in particular, been widely used for climate and environmental modelling. Such methods are particularly useful when the underlying science is encoded using deterministic simulators based, for example, on collections of deterministic differential equations. A few costly runs from such massive simulators are taken from certain designed scenarios and these are then used to frame judgements across the whole space of interest using either Gaussian processes anchored at the results of the runs \citep{Conti2009,Sanso2008} or Bayes Linear methods which use plausible continuity assumptions to interpolate expectations and their associated uncertainties over the whole domain of interest \citep{Williamson2011,Williamson2012}. These are all able to provide probabilistic outputs and so, in particular, the various moments of critical features of the problem we will later demonstrate  we will need for our IDSS.

One element all these types of model have in common is that they are not just arbitrary probability models. As well as being able to deliver various conditional probability statements  if queried, all are also able to deliver a rationale that lies behind the delivered numerical evaluations. This rationale can of course take one or several different forms: an underlying scientific justification, based on experimental, survey or observational evidence, simulation runs from detailed complex systems, carefully elicited judgements from respected experts in the field and so on. But, whatever form the justification might take, this therefore means that if challenged - for example by some external auditor or regulator - a robust defence of the probability statements can be given. For example, forensic DSSs concerning DNA evidence are based on established scientific theory and a plausible dependence model. Economic forecasts will be based on defensible models of the economy together with observational data, climate change forecasts on simulator runs coupled with emulators that interpolate these results more widely under plausible smoothness assumptions, BNs of ecological systems on carefully elicited structural relationships between its variables and their conditional probability tables, also usually themselves supported by observational and scientific data. These types of explanation are a vital component for any DSS making it a compelling, practical tool.

The point is therefore that for any sort of decision support it is not enough to demand only coherence in its formal sense. It is also essential to be able to provide - if called upon to do so - a narrative which defends the probability models and is plausible enough to encourage a third party scrutinising the model to at least suspend their beliefs and accept that the probability statements are plausible  \citep{Smith1996}. Statistically motivated systems with this property are sometimes called \emph{internalist} by philosophers: see e.g. \citet{Peterson2009}. All probabilistic models in this domain we are aware of are not simply black boxes but have this additional property. So this facility will also henceforth be assumed for any contributing component model.  

The diverse collection of different types of probability models -  each carrying its own supporting narrative -  are now available to support single panels of experts within a composite system. Some of these models may be very large, encompassing long vectors of explanatory observables and modelling complex relationships between them. Others might be entirely subjective and reflect the expert judgements of the panel in a probabilistic form. But in all cases, it is reasonable to request that a panel delivers a collection of outputs - typically various expectations of functions of explanatory conditioning variables - together with  the ability to supply a supporting narrative of the type discussed above. We will demonstrate that our methodology then determines, not only what these summary outputs should be, but also describes how they can be processed to provide a decision centre with a coherent and global picture of the process as a whole.

\subsection{What an IDSS does}
Given a suite of different models, like the ones reviewed in the previous section, an overarching probabilistic methodology needs to enable us to accommodate the diversity of information and its intrinsic uncertainty coming from these submodules into the system. The motivation of this paper is to determine when and how a supporting narrative can be composed around the component narratives discussed above that can then be used to explain to any outside auditor the rationale behind the choice of decision taken in both a formally justifiable and plausible way. To address this issue we start by looking at a small system and then gradually increase the size of applications so that the ideas can also be applied to the large domains of nuclear emergency management and UK food security applications reviewed in Section \ref{appl}.

Under certain hypotheses, and given a variety of structural assumptions, we are able to develop a methodology, similar to a standard Bayesian one, where decisions can be guaranteed to be coherent, i.e. expected utility maximising for some utility and probability distribution derived from individual but connected suites of models of the types discussed above, and defensible enough to support a composite narrative in a sense we will define precisely in Section \ref{sec:coherence}. The derived IDSS will then be able to fully support a subjective expected utility maximising crisis centre or policy making forum in a justifiable way and help it draw together all the evidence distributed across different sources whilst properly taking into account the strength of the evidence on which these judgements are made. We demonstrate in particular that  these properties are often implicit when, in a formal sense, the system is casual: an assumption we later argue is implicitly made when building real models.

The methodology developed in this paper then provides not only a framework for faithfully encoding all usable and informed expert judgements and data leading to scores for competing candidate policies but also an overarching narrative explaining the derivation of these scores. This narrative will be composed of a sequence of sub-narratives delivered by the particular relevant panels of experts. So, in this sense, it is based on best evidence. This then provides a platform around which a decision centre can discuss the evidence supporting one policy against another. On the basis of this platform, assessments can be discussed and revised where deemed necessary:  an interactive capability commonly recognised vital to any such IDSS.

We show that sufficient conditions under which such an interactive IDSS can be built are ones that lead to the system being distributive. By this we mean that it is coherent for each panel to autonomously focus only on its own field of expertise and update its beliefs about the domain under its jurisdiction when new evidence is introduced in the IDSS. We later show that we can often attain this property provided that the IDSS has an appropriate protocol guiding the nature and quality of the data input by each of its component systems. This distributivity property gives the added benefit that the expected utilities it needs can typically be calculated very quickly using algorithms analogous to fast propagation algorithms used in BNs. These algorithms are customised to an overarching agreed dependence structure across the system as a whole and have recently been discussed in  \citet{Leonelli2015}. They then constitute the inferential engine of an IDSS and make their outputs not only formal and transparent but also feasible to implement.

We start setting up this formal framework for the combination of panels' judgements and subsystems in Section \ref{sec:toy} and introduce a toy example to illustrate the challenges and opportunities presented by even very small networks of systems. We then prove in Section \ref{sec:coherence} some key results that enable us to address these new inferential challenges in  more complex settings. In particular, we derive a set of conditions which ensure an IDSS is coherent and a faithful expression of this composite process in a sense to be made explicit later. We also show how and when such a system can legitimately devolve judgements to domain experts so that the IDSS remains distributed and so  feasible as well as sound. In these settings both estimation and validation can be performed locally by the individual panels of experts contributing to the composite inference. In Sections \ref{sec:frameworks} and \ref{sec:MDM} we proceed to illustrate our methodologies as they might apply to a range of different overarching structures, including dynamic ones.

\section{Networks of probabilistic expert systems}
\label{sec:toy}
\subsection{Some technical structure}
Assume that the different components of a network of processes are evaluated and overseen by $m\in\mathbb{N}$ different panels of domain experts, $\{ G_{1},\dots, G_{m}\} $, and let $[m]=\{1,\dots,m\}$. Examples of such diverse panels operating various components of a network have already been given in Section \ref{appl}. Let $d\in \mathbb{D}$ denote a control strategy or policy a decision centre might adopt from a class $\mathbb{D}$ of available policies. We envisage that a large vector of random variables measures various features of an unfolding future. Henceforth denote the vector of these random vectors by $\boldsymbol{Y}=( \bm{Y}_{i}^\T)_{i\in[m]} $, where $\bm{Y}_{i}$ takes values in $\mathbb{Y}_{i}(d)$, $i\in[ m],\ d \in \mathbb{D}$. Often these random vectors will be indexed by time. Panel $G_{i}$ will be responsible for the output vector $\bm{Y}_i$, $i\in[m]$. The implicit (albeit virtual) owner of beliefs expressed in the system will be henceforth referred to as the \emph{supraBayesian} (SB).

For each $d\in \mathbb{D}$, each panel $G_{i}$ will be asked to give the SB various summaries of the probability distribution of the subvector $ \bm{Y}_{i}$ of $\bm{Y}$ over which $G_{i}$ has oversight, conditional on certain measurable functions taking values $\left\{ A_{i}( \bm{Y}):A_{i}\in \mathbb{A}_i\right\} $ where $\mathbb{A}_{i}$ could be null. So for example $\mathbb{A}_{i}$ could be the set of different possible combinations of levels of the covariates on which the vector $\bm{Y}_i$ might depend. The SB will need to process these \emph{necessary probabilistic features} provided by the different panels. She will then use these to calculate various statistics of a potential decision centre's \emph{reward} vector, some function of $\bm{Y}$. Using these features, the SB will then calculate her expected utility scores $\left\{ \bar{U}(d):d\in \mathbb{D}\right\} $ for each policy to which she might commit. Such scores will obviously be functions of the centre's utility $U(\bm{y},d)$, where $\bm{y}$ is an instantiation of $\bm{Y}$, drawn from some class of utilities $\mathbb{U}$, with any structural modelling assumptions and the probability statements provided by the individual panels. Note that, as illustrated in the RODOS example above, such utility functions will usually have several attributes as its arguments. With these scores the decision centre will then be able to identify a decision $d^{\ast }\in  \mathbb{D}$ with the associated highest score $\bar{U}(d^{\ast })$ together with other high scoring decisions for further comparison and discussion.

Ideally, for an IDSS to be defensible, it should endeavour to accommodate probabilistic information provided only by the most well-informed experts. For this to happen, it would make sense for different choices of decisions $d \in \mathbb{D}$ for each panel \emph{only} to donate probabilistic summaries associated with their own \textbf{particular domain of expertise}, and not their beliefs about the whole vector of components $\bm{Y}$.  Some conditions that lead to the necessity of this from a formal coherence viewpoint are given in Section \ref{sec:coherence}. 

Assume then that $G_{i}$ will be able, possibly with the use of their own probabilistic DSS, to perform probabilistic inference over its own particular domain of responsibility. Typically, in practice, the panel's chosen system will support decisions over much more complex scenarios than those concerned in the specific crisis management or policy forum our IDSS might be designed to inform. For this task, as discussed in more detail later, the integrating system will usually only need panels to deliver certain \textbf{distributional summaries} of the measurements $ \boldsymbol{Y}$ under each chosen decision. So for example a production module for a food security IDSS may be able to predict yield of a particular produce by farm. However to inform the flow of food in the system it would need only give its aggregate forecasts associated with produce as it arrives at market.  The SB will then use these summaries as her own, in a way we will outline later in Section \ref{sec:MDM}.

 So assume that, for $i\in[m]$, $G_{i}$  will be required to deliver to the IDSS belief summaries denoted by
\begin{equation*}
\Pi _{i}^{y}\triangleq \left\{ \Pi _{i}^{y}(A_i,d):d\in \mathbb{D}, \ A_i \in \mathbb{A}_i\right\}.
\end{equation*}
These summaries will typically be various \emph{expectations} of certain functions of $\boldsymbol{Y}_{i}$ \emph{conditional} on the values in $\mathbb{A}_{i}$ taken by some subvector of $\boldsymbol{Y}$ for each $d\in \mathbb{D}$. So, for instance, in a BN these might be the expected probabilities in a set of conditional probability tables and $\mathbb{A}_i=  \times_{j\in\Pi_i} \mathbb{Y}_{j}(d)$,  where the parents of $\boldsymbol{Y}_i$ are $\{\bm{Y}_{j}: j\in\Pi_i\}$ and $\Pi_i\subseteq[i{-1}]$. Note however that $\mathbb{A}_i$ does not necessarily need to be a product space. For example  in Section \ref{sec:tree} we discuss an IDSS for which our methods still apply but whose asymmetric structure does not admit such tabular form.

We show below that the belief summaries $\Pi_{i}^{y}$ can be determined once an overarching dependence framework has been agreed by all panellists in the system. It will contain \emph{only} those quantities $G_{i}$ will be required to deliver to the IDSS so that the IDSS is able to calculate its expected utility scores:  under quite general conditions often turning out to needing only to be  short vectors of expectations of certain functions conditional on the observations of certain events. This property - defined for the specific purpose of the IDSS - is central to being able to define a feasible IDSS even for large dynamic systems.


Henceforth we assume that all panellists make their inferences in a parametric or semi-parametric setting where $\boldsymbol{Y}$ is parametrised by $\boldsymbol{\theta }=\left( \bm{\theta}_i\right)_{i\in[m]} \in \Theta (d),\ d \in \mathbb{D}$. Here the parameter vector $\boldsymbol{\theta }_{i}$ parametrises  $G_i$'s relevant sample distributions,  $i\in[m]$. This may be infinite dimensional. When the parameter space of the system can be written as a product space, $\Theta(d) = \times_{i\in[m]}\Theta _{i}(d)$,  where $\Theta_i(d)$ is $G_i$'s parameter space, we say that panels are \emph{variationally independent} \citep[see][]{Dawid2001}. We henceforth assume this property holds. Were this not so then it would be necessary for a panel to state its beliefs about the value of $\boldsymbol{\theta} _{i}\in \Theta _{i}(d)$ in terms of parameters of the sample distributions of other panels.  We need to try to avoid this dependence so that it is possible for the system to be distributive. We show that, happily, many causal systems can be parametrised so that variational independence does indeed  hold.

In this parametric setting, for each $d\in \mathbb{D}$ that might be adopted, each panel $G_{i},$ 
$i\in[m]$,  has two quantities available to  them. The first is a set of \emph{sample summaries} over the future measurements for which they have responsibility 
\begin{equation*}
\Pi _{i}^{y|\theta }\triangleq \left\{ \Pi _{i}^{y|\theta }(\boldsymbol{\theta }_{i},A_i,d):\boldsymbol{\theta }_{i} \in \Theta _{i}(d), A_i \in \mathbb{A}_i, d\in \mathbb{D}\right\}.
\end{equation*}
These $\Pi _{i}^{y|\theta }$ might be the set of sample distributions associated with the predicted process 
\[
\left\{ f_i(\boldsymbol{Y}_{i}\ |\ A_i, \boldsymbol{\theta }_{i}):\boldsymbol{\theta }_{i} \in \Theta _{i}(d), A_i \in \mathbb{A}_i, d\in \mathbb{D}\right\}, 
\] 
where $\boldsymbol{\theta }_{i}\in \Theta _{i}(d)$ parametrises $f_i(\boldsymbol{Y}_{i}\ |\ A_i, \boldsymbol{\theta }_{i})$. For example, if $\boldsymbol{Y}$ were discrete and finite, then each panel might be asked to provide certain multi-way conditional probability tables over their subvector $\boldsymbol{Y}_{i}$, conditional on each $A_{i}\in \mathbb{A}_{i}$ and $d\in \mathbb{D}$. In this case $\boldsymbol{\theta }_{i}\in \Theta _{i}(d)$ would be the concatenated probabilities within all these tables for that chosen $d \in \mathbb{D}$. We have already noted that when a panel $G_i$ is supported by its own probabilistic system, then a typically much longer vector of parameters $\boldsymbol{\phi }_{i}\in \Phi _{i}(d)$ may be available to $G_i$  on which the panel is prepared to communicate uncertainty judgements. In this case, typically $\boldsymbol{\theta }_{i}$ will be a low-dimensional function of $\boldsymbol{\phi }_{i}$ capturing only $G_i$'s beliefs about features an IDSS needs. So, for example, a panel may have available a DSS designed to predict the health consequences of poisoning. But if an IDSS is designed to be used in an incident centre after a radiological accident, only the effect of poisoning from radiation, within the ranges of exposure of the accident and within ranges considered dangerous, will be needed for the decisions supported by the IDSS. Hence only the parameters of the margins of those features would need to appear in the $\boldsymbol{\theta }$ vector.

Second we will assume that each panel is able to express, and explain if questioned, its beliefs
\begin{equation*}
\Pi _{i}^{\theta }\triangleq \left\{ \Pi _{i}^{\theta }(A_i,d):A_i \in \mathbb{A}_i, d\in \mathbb{D}\right\}, 
\end{equation*}
about the parameters $\boldsymbol{\theta }_{i}\in \Theta _{i}(d)$, $d\in \mathbb{D}$, of its associated conditional distributions of $\boldsymbol{Y}_{i}\ |\ A_{i}\left( \boldsymbol{Y}\right)$. Most generally, \emph{panel beliefs} $\Pi _{i}^{\theta }(d)$ might be expressed in terms of \emph{panel densities} $\pi _{i}\left( \boldsymbol{\theta }_{i}\ |\ A_i, d\right)$. So in our example, this would be a joint probability distribution over all the probabilities specified in the conditional tables above. Note that a panel would not normally need to divulge how these judgements were made. For example, it would not need to show the details of any prior to posterior analyses \emph{unless} the panel were interrogated, for example, during emergency conferencing at the time of a crisis or by a regulator assuring the quality of the system before a crisis occurred.  In a parametric IDSS the vector of summaries $\Pi _{i}^{y}$, mentioned above, can obviously be calculated by $G_{i}$ through marginalisation. 
 
 Note that the inference performed by panel $G_i$ to provide their outputs is autonomous. What this means is that they should have available not only their outputs but also evidence about the statistical validity  of the structure and distributions they might define. This statistical justification - demonstrated by, for example, various diagnostic plots demonstrating the plausibility of modelling assumptions made within their component in the light of hard data evidence - can be assumed to be available on request, i.e. an audit trail behind each panel's probabilistic judgements is in place if the centre needs to query a panel's outputs. Such demonstrations will be henceforth assumed to be part of any supporting narrative, accessible through querying the component input.

Panel $G_i$ will of course use various data available to them to infer their distribution of $\bm{\theta}_i$. If they do this, they will typically perform this inference autonomously: i.e. without reference to the other panellists. Now, it is by no means automatic that such autonomous updating will be justified if $G_i$'s inferences are going to be inherited by the composite system. Examples \ref{Ex:2_5} and \ref{ex2.6} below give illustrations of when such autonomy is not formally justified. Later in the paper we determine sets of sufficient conditions when such delegation is formally possible. 
\begin{definition}
We call an IDSS \emph{distributed} if the SB's beliefs are functions of the autonomously calculated beliefs of the individual panels $G_1,\dots, G_m$. 
\end{definition}

\subsection{Common knowledge assumptions for an IDSS}
Let us begin by assuming that after a series of decision conferences \citep{French2009} held jointly across the panels, and electronic communications, stakeholders and users have all agreed the types of decisions the IDSS will support to a sufficient level of specificity and provided an agreed qualitative framework across all interested parties around which a quantitative framework can subsequently be built. To this purpose we assume three properties hold.
\begin{property}[Policy consensus]
All agree the class of decision rules $d\in \mathbb{D}$ whose efficacy might be examined by the IDSS.
\end{property}
\noindent This class of feasible policies considered will depend not only on what is logical, such as when various pieces of information are likely to become available, but also what might be acceptable and allowable, either legally or for other reasons. Again the choice of $\mathbb{D}$ will often be resolved using decision conferencing across panel representatives, users and stakeholders, as was the case during the Chernobyl project and the construction of RODOS discussed above. In the case of the county council policy analysis, the decision space $\mathbb{D}$ contains the different ways to legally implement central government cutbacks in the services provided to the needy and vulnerable in the county. Although we do not dwell on this point here, for an efficient and transparent system it is critical to customize this functionality carefully, so that the IDSS supports the real decision-making of the centre.
\begin{property}[Utility consensus]
All agree on the class $\mathbb{U}$ of utility functions supported by the IDSS.
\end{property}
\noindent In the complex multivariate settings we address here, a utility function $U(\bm{y},d)$ needs to entertain certain types of preferential independence across its various attributes, where these attributes will need a priori to be agreed. In the case of RODOS these were usually measures of health consequences, public acceptability and cost of each possible countermeasure policy taken over space and time. In both our illustrative examples the family of utility functions is simply one of value independence \citep{Keeney1993} although this is certainly not a necessary condition for our methods to apply \citep{Leonelli2015}.
\begin{property}[Structural consensus]
All agree the variables $\bm{Y}$ defining the process, where, for each $d\in \mathbb{D}$, each $U\in \mathbb{U}$ is a function of $\bm{Y}$, together with a set of qualitative statements about the dependence between various functions of $\bm{Y}$ and $\bm{\theta}$. Call this set of assumptions the \textbf{structural consensus set} and denote this by $\mathbb{S}$.
\end{property}
\noindent This last consensus might be expressible through an agreement about the validity of a particular graphical or conditional independence structure across not only the distribution of $\left( \bm{Y}\ |\ \bm{\theta }\right) $, but also the one of $\bm{\theta }$ \citep{Smith1996}.  This is then hard-wired into the IDSS. These types of assumptions are often complex, so we defer their discussion to later in the paper and examples of these, including those used in our illustrative applications, will be given in Sections \ref{sec:frameworks} and \ref{sec:MDM} below.  Other information that might be included in $\mathbb{S}$ could be a consensus about certain structural zeros or known logical constraints arising a shared understanding of the meaning of certain variables. 

\begin{definition}
Call the set of common knowledge assumptions shared by all panels and which contains the union of the utility, policy and structural consensus $\left( \mathbb{U},\mathbb{D},\mathbb{S}\right) $ the \textbf{\emph{CK class}}.
\end{definition}
\noindent Technically we can think of the CK class as the \emph{qualitative} beliefs that are shared as common knowledge by all the panel members and potential users, who all know they know, and so on. The CK class will be the foundation on which all inference within the IDSS will take place. Note that this class will depend not only on the domain and needs of users of the system, but also on the constitution and knowledge bases of the panels.
\begin{definition}
Call an IDSS \textbf{\emph{adequate}} for a CK class $\left( \mathbb{U},\mathbb{D},\mathbb{S}\right) $ when the SB can unambiguously calculate her expected utility score $\bar{U}(d)$ for any decision $d\in \mathbb{D}$ and  any utility function $U\in \mathbb{U}$ from the panel marginal inputs $\Pi _{i}^{y}$ provided to her by   $G_{i}$, $i\in[m].$ 
\end{definition}
\noindent An adequate IDSS will be able to derive a unique score for each $d\in\mathbb{D}$ on the basis of the panels' inputs. An IDSS clearly cannot be fully functional unless it has this property. Note that it should be immediate from the formulae of a given probabilistic composition to calculate these expectations whether or not the system is adequate. We illustrate such formulae later in the paper. 

To calculate $\left\{\bar{U}(d):d\in \mathbb{D}\right\}$ the SB will need, together with the CK class,  enough probabilistic information to compute the expectations of the corresponding utilities. At worst this might need to be the full distribution of $\bm{Y}$. More commonly, for typical choices of $\mathbb{U}$, all that might be needed is the distribution of the margins on certain specific functions of $\bm{Y}$ or simply some summaries such as a selection of its moments, again indexed by $d$.

To be defensible - in the sense that the explanations of the appropriateness of its delivered outputs provided by panels can also be legitimately adopted by the IDSS - a parametric IDSS needs another property.

\begin{definition}
Call an IDSS \textbf{\emph{sound}} for a CK class $\left( \mathbb{U},\mathbb{D}, \mathbb{S}\right) $ if it is adequate and, by adopting the structural consensus, the SB would be able to admit coherently all the assessments $\Pi_{i}^{y|\theta}$ and $\Pi _{i}^{\theta}$ (and hence $\Pi _{i}^{y}$) as her own, the SB's underlying beliefs about a domain overseen by a panel $G_{i}$ being $\{ \Pi _{i}^{y|\theta },\Pi _{i}^{\theta}\}$, $i\in[m]$.
\end{definition}
\noindent In Theorem \ref{theo:gold} we give a set of necessary conditions that in general guarantee the soundness of an IDSS. We note that in a surprising array of different circumstances an IDSS can be designed so that it is sound.

 A sound IDSS does not necessarily need to embody the \emph{full} beliefs held by all panel members and  based on the totality of their own individual evidence. This would often be inappropriate for a shared belief system, whose outputs will need to be defensible. For example, the evidence used to form the subjective judgements of individual panel members, although compelling to them, may derive from poorly designed experiments or simply be anecdotal. Because such information could not be robustly defended, it might not be possible for the centre to adopt it. So for example, an as yet unpublished observational study on those exposed to radiation after  Chernobyl might strongly indicate that the effects of increase of cancers, commonly predicted, have been grossly exaggerated. Although the relevant panel might find this strongly compelling, it might not be appropriate to input this information into a common knowledge system because the study has yet to be adopted generally by the scientific community.

 The sound IDSS does, however, present a defensible and conservative position all panellist should be happy to communicate and provide a benchmark for further discussion. In this sense the beliefs expressed in an IDSS are analogous to a pignistic belief system \citep{Smets2005}: the best legitimate belief statements that can be made if the centre is called to act under uncertainty in a coherent and justifiable way. 

To illustrate how these properties might apply even in a trivial setting we consider the following simple example. 

\begin{example}
\label{ex:bin}
Consider the simplest possible scenario where  $m=2$ and the CK class specifies that $\bm{Y}= (Y_{1},Y_{2})^{\textnormal{T}}$, where both $Y_{1}$ and $Y_{2}$ are binary. Here the random variable $Y_{1}$ is an indicator of whether or not a food stuff has become poisonous and  $Y_{2}$ is an indicator of whether or not sufficient quality controls are in place to ensure that any contamination is detected before the food is distributed to the public. A family of sample distributions $\Pi_{1}^{y|\theta }$ given by panel $G_{1}$, expert in the processes that might lead to poisoning, is saturated so that $\theta_{1}\triangleq P(Y_{1}=1)$. Panel $G_{2}$, consisting of experts with a good knowledge of quality control systems, has beliefs about the probabilities 
\[
\begin{array}{cc}
\theta _{20}\triangleq P(Y_{2}=1\ |\ Y_{1}=0),\;\;\;&\;\;\;\theta_{21}\triangleq P(Y_{2}=1\ |\ Y_{1}=1).
\end{array}
\]
Write $\bm{\theta }_{2}\triangleq \left(\theta _{20},\theta _{21}\right)^{\textnormal{T}}$. If within the CK class $\mathbb{U}$ is an arbitrary utility on $\bm{Y}$, then for adequacy the SB will need to be able to calculate her expected joint probability table of $\bm{Y}$, i.e. the expectations $\bar{\bm{\mu }}\triangleq \left( \bar{\mu }_{00},\bar{\mu }_{01},\bar{\mu }_{10},\bar{\mu }_{11}\right) $ of $\bar{\bm{\theta }}\triangleq \left( \bar{\theta }_{00},\bar{\theta }_{01},\bar{\theta }_{10},\bar{\theta }_{11}\right) $, where by definition 
\begin{equation*}
\begin{array}{cccc}
\bar{\theta }_{00}=(1-\theta _{1})(1-\theta _{20}), & \bar{\theta }_{01}=(1-\theta _{1})\theta _{20}, & \bar{\theta }_{10}=\theta_{1}(1-\theta _{21}), & \bar{\theta }_{11}=\theta _{1}\theta _{21}.
\end{array}
\label{binaryexdep}
\end{equation*}
Assume that prior panel independence is within the CK class: i.e. there is a consensus between the members of the two panels that $\bm{\theta }_{2}$ is independent of $\theta _{1}$. Then writing 
\begin{equation*}
\begin{array}{cc}
\mu _{1}(d)\triangleq \mathbb{E}\left( \theta _{1}\ |\ d\right) , & \bm{\mu }_{2}(d)\triangleq \left( \mu _{20}(d),\mu _{21}(d)\right) \triangleq\left(\mathbb{E(}\theta _{20}\ |\ d),\mathbb{E(}\theta _{10}\ |\ d)\right),
\end{array}
\end{equation*}
we would have that 
\begin{equation}
\begin{array}{ll}
\bar{\mu }_{00}(d)=(1-\mu _{1}\left( d)\right) (1-\mu _{20}\left(d\right) ), & \bar{\mu }_{01}(d)=(1-\mu _{1}(d))\mu _{20}(d), \\ \bar{\mu }_{10}(d)=\mu _{1}(d)(1-\mu _{21}(d)), & \bar{\mu }_{11}(d)=\mu _{1}(d)\mu _{21}(d).
\end{array}
\label{exp uti}
\end{equation}
Suppose panels $G_{1}$ and $G_2$ are able to calculate, respectively 
\begin{align*}
\begin{array}{ccc}
\Pi _{1}^{\theta }=\left\{ \mu _{1}(d):d\in \mathbb{D}\right\},&\ \text{and}\ &
\Pi _{2}^{\theta }=\left\{ \bm{\mu }_{2}(d):d\in \mathbb{D}\right\}.
\end{array}
\end{align*}
Then the IDSS is \emph{adequate}.  Because of the properties the expectations, in this case the  belief summaries $G_i$ need to deliver are simply $\Pi^y_i=\Pi^\theta_i,\ i\in[2].$  This provides the SB with all the information she needs to calculate all her expected utility functions using the formulae in (\ref{exp uti}). The IDSS is also \emph{sound} since these inputs are consistent with the probabilistic beliefs of anyone with any probability model over $(\bm{Y},\theta _{1},\bm{\theta }_{2})$  who believed the agreed prior panel independence assumptions and held the expectations given above.
\end{example}

Note that, for any $d\in \mathbb{D},$ it is not a trivial condition that the SB can make the calculations she needs in terms of $\Pi _{i}^{y}$, $i\in[2]$. For example, if instead of providing its beliefs about the conditional probabilities, panel $G_2$ provided its beliefs about the margin of $Y_2,$ the marginal joint distribution of $\bm{Y}=(Y_{1},Y_{2})$ would not then  be fully recoverable since we have nothing from which to derive, for example, the covariance between $\theta_{1}$ and $\bm{\theta }_{2}$ which is needed to calculate the covariance between $Y_{1}$ and $Y_{2}$.  So, if structuring of the process is not performed beforehand, then post-hoc combinations of outputs from panels' models may not be formally possible.

\subsection{An illustration of some of the inferential challenges}
\label{sec:inf}
It is convenient at this stage to use another very simple example to illustrate which statistics need to be communicated by panels to an IDSS.

\begin{example}
\label{ex:ill}
Assume a CK class gives $\left( Y_{1},Y_{2}\right) $ the same meaning and sample space as in Example \ref{ex:bin}. However add to the CK class the additional structural assumption that $Y_{2}\independent Y_{1}\;|\;(\theta_{1},\theta _{2})$ whatever decision $d\in \mathbb{D}$ is made. Thus, once the probabilities of these events were known, it is generally accepted that learning that contamination had been introduced would not affect our judgements about the efficacy of the quality control regime. Suppose $G_{i}$ delivers the set of beta distributions $\textnormal{Be}(\alpha _{i}(d),\beta _{i}(d))$ for $\theta _{i}=P(Y_{i}=1)$, $d\in \mathbb{D}$, $i\in[2]$. Note that because of the structural assumption above, in the notation used in Example \ref{ex:bin}, $\theta_{2}=\theta _{20}=\theta _{21}$. 
Consider two possible CK classes: where a decision centre is known to draw its utilities $U_{i}\in \mathbb{U}_{i}$, $i\in[2]$, from one of the families below 
\begin{align*}
U_1(y_{1},y_{2},d) &=a_1(d)+b_{11}(d)y_{1}+b_{12}(d)y_{2},	&&\text{if }  U_1 \in \mathbb{U}_1, \\
U_2(y_{1},y_{2},d) &=a_2(d)+b_{2}(d)w, 									&&\text{if } U_2 \in \mathbb{U}_2,
\end{align*}
and $W\triangleq Y_{1}Y_{2}$ is an indicator of whether the public is exposed to the contamination, with $a_{1}(d), a_{2}(d)\in\mathbb{R}$ and $b_{11}(d),b_{12}(d), b_2(d)\in\mathbb{R}_{>0}$ for all $d\in\mathbb{D}$. If $\mathbb{U}_{1}$ is in the CK class then the SB needs only that $G_{i}$ supplies its mean $\mu _{i}(d)=\alpha_{i}(d)(\alpha _{i}(d)+\beta _{i}(d))^{-1}$ of $\theta _{i}$, $i\in[2]$, as a function of the decision $d\in \mathbb{D}$ taken: a simple one-dimensional summary. However if, instead,  $\mathbb{U}_{2}$ is in the CK class then the SB needs to be able to calculate 
\begin{equation*}
\mathbb{E(}W\ |\ d)=\mathbb{E}\left( \theta_{1}\theta _{2}\ |\ d\right),
\end{equation*}
for each $d\in \mathbb{D}.$ It is easily checked that the above panel summaries would no longer necessarily be adequate if $\mathbb{U}_{2}$ was in the CK class unless further assumptions were added.  Explicitly, the SB would also need to add to the CK class a global independence assumption $\theta _{1}\independent \theta _{2}$. If this was done then the distribution of $\theta =\theta_{1}\theta _{2}$ would be recoverable from the panels' expectations and thus
\begin{equation*}
 \mathbb{E}(W\ |\ d)=\mathbb{E}(\theta\ |\ d)=\mu _{1}(d)\mu _{2}(d),
\end{equation*}
would be well defined. So the IDSS would be adequate. 
\end{example}

To be feasible and of enduring usefulness, it is usually necessary to require that the IDSS is distributive so that panels can autonomously update their probabilistic beliefs about their area of responsibility as they receive new information. 
\begin{example} 
\label{ex:2.3}
Assume a random vector $\left( \bm{X}_{1},\bm{X}_{2}\right) $ is sampled from the same population as $\left( Y_{1},Y_{2}\right) $ in the model of Example \ref{ex:ill} and that, for each $d\in \mathbb{D}$, $\theta _{1}\independent \theta_{2}$  is in the CK class. Each panel $G_{i}$ next refines its probabilistic assessments by observing its own separate randomly sampled populations, $\bm{x}_{i}$, concerning $\theta_i$ alone, and then updates its parameter densities, given each $d\in \mathbb{D}$, from $\pi _{i}(\theta _{i}\ |\ d)$ to $\pi_{i}(\theta _{i}\ |\ \bm{x}_{i},d)$, $i\in[2]$. In this case,  the two panels need to deliver only their  posterior means $\mu _{i}^{\ast }(d,\bm{x}_{i})$, $i\in[2]$, $d\in \mathbb{D}$. The SB can then act coherently. By adopting all these beliefs as her own, she will act as if she had sight of all the available information and had processed this information herself. The IDSS is therefore sound and distributed. 
\end{example}

Note, however, that in the example above  the global independence assumption is critical for this distributivity property to hold. 
\begin{example}
\label{ex:2.4}
Suppose that $\theta_2 $ is not independent of $\theta_1$ so that  $\pi (\theta _{2}\ |\ \theta _{1})$  needs to be a function of $\theta _{1}$ for at least some $d \in \mathbb{D}'\subset\mathbb{D}$. Then, in the notation of Example \ref{ex:2.3}, for these $d \in \mathbb{D}'$,
\begin{equation*}
\pi _{2}(\theta _{2}\ |\ \boldsymbol{x}_{1},\boldsymbol{x}_{2}, d)=\int_{0}^{1}\pi _{2}(\theta _{2}\ |\ \theta_{1},\boldsymbol{x}_2, d)\pi_1(\theta _{1}\ |\ \boldsymbol{x}_1,d)\dr\theta _{2}.
\end{equation*}
where the prior dependence of $\theta_2$ on $\theta_1$ induces a dependence of $\theta_2$ on $\boldsymbol{x}_1$. So, in particular, $\mu _{2}^{\ast }(d,\bm{x}_{1},\bm{x}_{2})\neq \mu _{2}^{\ast }(d,\bm{x}_{2})$ in general. Therefore, by devolving inference to the two panels who learn autonomously,  the SB will \emph{not} be acting as a single Bayesian would by using $\mu_i^*(d,\bm{x}_i)$, $i\in[2]$. It follows that the system is no longer sound, although when supporting evidence remains unseen the SB will appear to act coherently. The explanation of her inferences can no longer be devolved to a single panel and so difficult to defend.  She will, implicitly, be assuming that $\theta_1 \independent \theta_2$, which is contrary to the reasoning $G_2$ would want to provide.
\end{example} 

Perhaps of even more importance, is to note that \emph{even} if global independence is justified \emph{a priori}, the assumption that data collected by the two panels and individually used to adjust their beliefs does not inform both parameters is also a critical one. There are two important special cases we examine below which give a flavour of this difficulty and illustrate why it is important to construct panels not only on the basis of domains of expertise but also whose composition matches, as far as possible, domains over which supporting vectors of measurements exist.

\begin{example}
\label{Ex:2_5}
Continuing from the last example, suppose that $G_{1}$ and $G_{2}$ both see their respective margin concerning the experiment in Table \ref{experiment}, where $100$ units from a population are randomly sampled, and each uses this experiment to update its respective marginal distribution on $\theta _{i}$ for a particular value of $d\in \mathbb{D}$, $i\in[2]$.
\begin{table}
\[
\begin{array}{ccccc}
Y_{1}\backslash Y_{2} & 0 & 1 &  &  \\ 
0 & 5 & 45 & 50 & n-x_{1} \\ 
1 & 45 & 5 & 50 & x_{1} \\ 
& 50 & 50 & 100 &  \\ 
& n-x_{2} & x_{2} &  &
\end{array}
\]
\caption{Experiment of Example \ref{Ex:2_5}.\label{experiment}}
\end{table} 
Then, if both began with a prior symmetric about $0.5$, each would believe that 
\begin{equation*}
\mu _{1}^{\ast }(d,x_{1})=\mu _{2}^{\ast }(d,x_{2})=0.5.
\end{equation*}
So were $\theta _{1}\independent \theta _{2}$ in the CK class, the utility function $U\in\mathbb{U}_2$ and data the individual panels used was na{\"i}vely restricted to the relevant margin, the IDSS would assign 
\[\bar{\mu }(d,x_1, x_2)=\mu_1^{\ast}(d, x_1) \mu_2^{\ast}(d, x_2)=0.25.
\]
Note that this inference contrasts with inferences the SB would make on seeing the whole table and assuming $\theta _{1}\independent \theta _{2}$ \emph{a priori}. With a fairly uninformative prior on the two margins, her posterior mean of $\theta_1 \theta_2$ would be approximately $0.05$, i.e. five times smaller than the expectation calculated above. 
\end{example}
 
So, when incorporating joint data of this type, it is not easy to preserve soundness. Experiments measuring a function of the variables with error - even when these functions relate directly to the predictions at hand - can induce similar difficulties. 

\begin{example} 
\label{ex2.6}
Continuing the setting in the example above, suppose it is only possible to see the table of randomly sampled counts associated with $W\triangleq Y_{1}Y_{2}$, i.e. the number of foodstuffs that have poisoned someone. Suppose the $100$ individuals in Table \ref{experiment} could be thought of as having been drawn from a Binomial experiment with $x_{w}=5$ values of $W=1$ within the sample. Suppose the SB uses this information directly: for example by introducing a uniform prior on $\theta_w \triangleq P(W=1)$. This would lead the SB to have a posterior mean of $\bar{\mu }(d,x_{w})\simeq \bar{\mu }(d,x_{1}\bm{,}x_{2})= 0.05$. However, observations have induced a dependence across  $\theta _{1}$ and $\theta _{2}$: the global independence assumption is no longer formally valid and, if we plan to demand that the IDSS is sound, the future distributivity of the system will be destroyed if this data is accommodated, and so frustrate future calculations. 
\end{example}

So, we have illustrated that, even in the simplest of networks, considerable care needs to be exercised before an IDSS can be expected to work reliably. Because of the simplicity of the examples above only means needed to be delivered. However once we move away from the case of two binary variables this is, generally, not the case. Nevertheless, we will see later that, as we increase the complexity of our problems, often each panel will need only to provide certain additional lower-order moments. In the next section we will prove some conditions which ensure our IDSS, however large, will be sound.  We will also discuss protocols for admitting data, that might be adopted by panels, designed to avoid the sorts of issues illustrated in the last two examples.

\section{Coherence and the IDSS}
\label{sec:coherence}
\subsection{Conditions to ensure a sound and distributive IDSS}
\subsubsection{Information and Admissibility Protocols}
\label{ssec:admis}
Suppose the IDSS is dynamically presented with a large amount of new information as time progresses.  In practice, within the totality of information conceivably available to panellists at time $t$, $J_{+}^{t}$, usually only a subset - the \emph{admissible evidence} - will be of sufficient quality and have suitable form to be integrated into an IDSS. The sorts of information excluded or delayed admittance might include evidence whose relevance is ambiguous or of a type which might introduce insurmountable computational challenges to an IDSS. An admissibility protocol is therefore needed to define the admissible evidence so that inferences made using the IDSS can be defended and feasibly and formally analysed within a required time frame. 

Let $I_{0}^{t}$ denote all the admissible evidence which is common knowledge to all panel members at time $t$. Let $I_{ij}^{t}$ denote the subset of this admissible evidence panel $G_i$ would use at time $t$ if acting autonomously to assess their beliefs about $\boldsymbol{\theta }_{j}$,  $i,j\in[m]$. Define the \emph{admissible evidence} as $I_{+}^{t} \triangleq \left\{ I_{ij}^{t}:i,j\in[m]\right\}$  and let $I_{\ast}^{t}\triangleq \left\{ I_{ii}^{t}:i\in[m]\right\}$ be the subset of the admissible evidence each panel $G_i$ would use to update $\bm{\theta}_i$, $i\in[m]$.

Of course, that there exist relevant protocols for the selection of good quality evidence for decision support is often assumed even within single agent systems, however, its explicit statement is frequently omitted. A notable exception is admissibility protocols for evidence concerning medical treatment where the Cochrane reviews are considered to be the gold standard in decision support \citep{Higgins2008}. Their purpose is to pare away information which might be ambiguous and so potentially distort inference, through a trusted set of principles relevant to the domain. This may seem restrictive, but in practical applications the need to be selective about experiments that can provide evidence of an acceptable quality before formally committing to policy - so that their adoption can be robustly defended - is universally acknowledged. The corresponding beliefs expressed within any support tool are therefore often, by their nature, conservative. Formal adoption in $I^t_+$ of evidence whose interpretation might be contentious should be avoided whenever possible.  Such information is best used as a supplement to the sound inference rather than being integrated into it. Note, however, that information in $J_{+}^{t}\backslash I_{+}^{t}$ can also be used formally by users and panellists to provide
diagnostic checks of the inference using $I_{+}^{t}$ alone. 

The demands for such an admissibility protocol of an IDSS are even more important than for standard single agent decision support, because of its collective structure. So here we assume that panels, both individually \emph{and} corporately, will agree an appropriate protocol for selecting suitable experimental evidence in line with good practices, mirroring Cochrane
reviews in ways relevant to their domain. However, one additional requirement is needed in this setting: the chosen admissibility protocol must also ensure that an IDSS remains distributed over time, for we have already argued that, if this is not the case, then the output of the IDSS is either dependent on arbitrary assumptions and difficult to calculate or, if distributivity is forced, will become
incoherent. We now explore the properties of a candidate set of admissible evidence $I_{+}^{t}$ that lead to an IDSS being both defensible and coherent.

\subsubsection{Conditional independence in CK parametric models}

We begin by assuming that the collective, as represented by the SB, is happy for its inferences to obey the (qualitative) semi-graphoid axioms given in the appendix of this paper \citep{Pearl1988,Smith2010}. These general properties are widely accepted as appropriate for reasoning about evidence when irrelevance statements are read as conditional independence statements: in particular Bayesian systems always respect these properties \citep{Dawid2001, Studeny2006}. Irrelevance statements can also be expressed in common language and so are more likely to form part of the common knowledge shared by a set of panellists \citep{Smith1996}. We later investigate how these ideas translate when all panellists are fully Bayesian.

In this setting, it is useful to recall that one useful definition of a parameter $\boldsymbol{\theta }\in \Theta $ in any parametric model can be phrased in terms of irrelevance. Explicitly, it might be common knowledge that  everything relevant to the future random vector $\boldsymbol{Y}$ that might be chosen from the totality $J_{+}^t$ of past information available at time $t$ is embodied in what we have learned about a parameter $\boldsymbol{\theta }\in \Theta $ and $I_+^t$. This would then imply the irrelevance statement that for all $d \in \mathbb{D}$
\begin{equation*}
\boldsymbol{Y}\independent J_{+}^t\ |\ \boldsymbol{\theta },I_{+}^t, d.
\end{equation*}
Here $I_{+}^t$ in particular contains the structural assumptions of the model, $\mathbb{S}$, and the sample distributions delivered by each panel. So, the SB's parameter vector is a concatenation of subvectors $\boldsymbol{\theta} =\left( \boldsymbol{\theta }_{i}\right)_{i\in[m]} $, where $G_{i}$ parametrises their delivered sample distributions by $\boldsymbol{\theta }_{i}$ , $i\in[m]$.

Now suppose it is common knowledge that the expected utilities $\bar{U}(d\ |\ I_{+}^t)$, $d\in \mathbb{D}$, posterior to observing $I^t_+$ and used by the SB to score her possible options can all be written in the form
\begin{equation*}
\bar{U}(d\ |\ I_{+}^t)\triangleq \mathbb{E}_{\bm{Y},\bm{\theta}|I_{+}^t}\left( \tau (\boldsymbol{Y},\boldsymbol{\theta}, d)\right), 
\end{equation*}
for some function $\tau (\boldsymbol{Y},\boldsymbol{\theta}, d)$. Further suppose it is common knowledge that $\tau (\boldsymbol{Y},\boldsymbol{\theta} ,d)$ has the form

\begin{equation*}
\tau (\boldsymbol{Y},\boldsymbol{\theta}, d)=\sum_{B\subseteq [m]}\prod\limits_{i\in B}\tau _{i,B}(\boldsymbol{Y},d)\rho_{i,B}(\boldsymbol{
\theta }_{i},\boldsymbol{Y}, d)
\end{equation*}
where, for $i\in[m]$,  the set 
$
\left\{ \tau _{i,B}(\boldsymbol{Y},d):B\subseteq [m],i\in B,d\in 
\mathbb{D},\boldsymbol{Y}\in \mathbb{Y}(d)\right\}
$
is known by $G_{i}$: a condition that encompasses many common classes of model \citep{Leonelli2015}. We will also see later that assuming most common structural frameworks as common knowledge, then under utility, policy and structural consensus the functions $\rho_{i,B}(\boldsymbol{\theta }_{i},\boldsymbol{Y},d)$ are simple functions of $ \Pi _{i}^{y|\theta }$, $i\in[m] $. 

\begin{definition}
Say that a CK class of an IDSS exhibits \textbf{\emph{panel independence} }to $
\{I_{+}^{t},d\}$ at time $t$ iff the SB believes that under any policy $d\in \mathbb{D}$ 
\begin{equation*}
\independent _{i\in[m]}\boldsymbol{\theta }_{i}\ |\ I_{+}^{t},d.
\end{equation*}
\end{definition}
\noindent Then if the IDSS exhibits panel independence, the expected utilities $\bar{U}(d)$ can be calculated using the formula

\begin{equation*}
\bar{U}(d\ |\ I_{+}^t)=\sum_{B\subseteq [m]}\mathbb{E}_{\boldsymbol{Y} |I^t_{+}}\left( \prod\limits_{i\in B}\tau _{i,B}(\boldsymbol{Y},d)\bar{\rho}_{_{i,B}}(\boldsymbol{Y},d)\right) ,
\end{equation*}
where $
\bar{\rho}_{_{i,B}}(\boldsymbol{Y},d)\triangleq \mathbb{E}_{\boldsymbol{\theta }_{i} | I^t_{+}}\left( \rho_{_{i,B}}(\boldsymbol{\theta }_{i},\boldsymbol{Y},d)\right)$. Now the key point here is that in this case the set of expectations 
\[
\bar{\rho}_{i}\triangleq \left\{ \bar{\rho}_{_{i,B}}(\boldsymbol{Y},d):B\subseteq [m],i\in B,d\in \mathbb{D},\boldsymbol{Y}\in \mathbb{Y}(d)\right\}, 
\] 
can be calculated \emph{locally} by panel $G_{i}$, $\ i\in[m]$. So all panels, and therefore the SB, can agree that a sufficient condition for  a system to remain distributed and adequate over all time is panel independence. So we next investigate conditions that will ensure that panel independence holds.

\subsubsection{Panel Independence and Common Knowledge}
We now define four properties to add to the CK class to ensure the soundness of an IDSS under a given admissibility protocol. Let $\bm{\theta }_{i^{-}}\triangleq \left( \theta _{j}\right)_{j\in[m]\setminus \{i\}}$.

\begin{definition}
\label{def:cond}
Say that a CK class of an IDSS is \textbf{\emph{delegable}} at time $t$ if for any
possible choice of policy $d\in \mathbb{D}$ and for $i\in[m] $ there
is a consensus that for all $\boldsymbol{\theta} \in \boldsymbol{\Theta}(d)$
\begin{equation}
I_{+}^{t}\independent \boldsymbol{\theta }\ |\ I_{0}^{t},I_{\ast }^{t},d
\label{delegatable}
\end{equation}
\textbf{\emph{separately informed}} at time $t$ if 
\begin{equation}
I_{ii}^{t}\independent \boldsymbol{\theta }_{i^{-}}\ |\ I_{0}^{t},\boldsymbol{\theta }
_{i},d  \label{sep inform}
\end{equation}
\textbf{\emph{cutting}} at time $t$ if 
\begin{equation}
I_{\ast }^{t}\independent \boldsymbol{\theta }_{i}\ |\ I_{0}^{t},I_{ii}^{t},
\boldsymbol{\theta }_{i^{-}},d  \label{cutting}
\end{equation}
\textbf{\emph{commonly separated}} at time $t$ if 
\begin{equation}
\independent _{i\in[m]}\boldsymbol{\theta }_{i}\ |\ I_{0}^{t},d  \label{commonly sep}
\end{equation}
\end{definition}
 An IDSS is delegable at time $t$ when it is in the CK class that, for any choice of future policy $d\in \mathbb{D}$, the totality of admissible evidence $I_{+}^{t}$ fed into the IDSS is the union of the evidence $I_{0}^{t}$ shared by all panels plus the aggregate of individual admissible evidence $I_{\ast}^{t}$ each panel has about its own particular domain of expertise at time $t$. Note that if the panel members are working collaboratively rather than competitively then this condition might be ensured through adopting a protocol where if one panel has new evidence which they think might inform another then they will immediately pass this on to that panel for appropriate accommodation: see below. Alternatively the protocol could itself simply demand that $I_{+}^{t}=\left\{I_{0}^{t},I_{\ast }^{t}\right\} $.

The next two assumptions then allow us to perform inference in the distributed way we will develop later. When a system is separately informed, pieces of evidence $G_{i}$ might collect individually \emph{will not be informative about the other parameters in the IDSS} owned by other panellists once the domain experts' evidence has been fed in. When a system is cutting, once $\left\{ I^t_{0},I_{ii}^{t}\right\} $ is known, no panel believes that another panel $G_{j}$ has used any information that $ G_{i}$ might also want to use to adjust its beliefs about $\boldsymbol{ \theta }_{i}$. This captures what we might mean when we call a panel `expert' over a particular domain. So, for example, to accommodate a piece of evidence, $G_{i}$ might first have needed to marginalise out a parameter in $\boldsymbol{\theta }_{j}$ because its sample distribution depended on this component. In this case the IDSS would not be cutting: this experiment told the SB not only about $\boldsymbol{\theta }_{i}$ but also $\boldsymbol{\theta }_{j}$. Formally, the assessment of these two parameters could then become dependent on each other \emph{a posteriori}, as illustrated in the last section. In practice the protocol would demand, to satisfy the separately informed condition, that a new piece of information would only be added to $ I_{ii}^{t}$  by $G_{i}$ if the strength of evidence it provided about $\boldsymbol{\theta }_{i}$ would not depend on $\boldsymbol{\theta }_{i^{-}}$. We see later that for a variety of overarching structure much information available to $G_{i}$ would satisfy this demand.  Typically this condition is broken with the loss of ancestral sampling in a BN  \citep[see, for example][]{Smith2010}.

When parameters are commonly separated all the \emph{information that everyone shares separates the parameters in the system}. Suppose all panels were constituted by the same people, the overarching system was a BN and the panels consisted of those deciding on the parameters $\boldsymbol{\theta }_{i}$ of the density of each vertex $Y_i$,  i.e. those defining the distribution of each variable conditional on its parents. Then this would reduce to the condition that global independence held at time $t$ \citep{Cowell1999}. From its proof, panel independence can actually be seen to be a consequence of the other properties we prove in Theorem \ref{theo:gold} below.

Now assume that the four conditions in Definition \ref{def:cond} all lie in $\mathbb{S}$. The following result, analogous to \citet{Goldstein1996} which concerned the use of linear Bayes methods and a single agent, can now be proved. 

\begin{theorem}
Suppose an IDSS for a CK class $\left( \mathbb{U},\mathbb{D},\mathbb{S}\right) $ is adequate where $\mathbb{U}$ and $\mathbb{D}$ are arbitrary and $\mathbb{S}$ includes the consensus that the IDSS is delegable, separately informed, cutting and commonly separated at time $t$. Then it will also be sound and distributed at time $t.$ Furthermore it is common knowledge that the SB's beliefs about each panel's parameter vector $\boldsymbol{\theta }_{i}$ are the same as those of the corresponding expert panel $G_{i}$, $i\in[m]$, for all $d\in \mathbb{D}$ and at any time $t \geq 0$. 
\label{theo:gold}
\end{theorem}

See Appendix \ref{gold} for a proof of this result. Certainly the conditions required in Theorem \ref{theo:gold}, above, are in no sense automatic.  We will show that, nevertheless, they are satisfied by a very diverse collection of models and information sets. So for example suppose that within our food security example $G_1$ was panel with expertise in food production and using a piece of probabilistic software for a model parametrized by $\bm{\theta}_1$. Suppose on the other
hand that $G_2$ was a panel with expertise in the effect of nutrition on a child's educational
attainment. Suppose that this panel based its judgements on a regression
model of such attainment on various food production indicator covariates parametrized by $\bm{\theta}_2$. Then delegability
would be the assumption that the composition of domain information
available to these two individuals covered the two areas - i.e. that both panels
could be thought of as expert. The separately and cutting informed hypotheses
asserts that neither expert has information available to them that they will use
in their own assessments of their own parameters that could also usefully be
used by other panelists. The commonly separated hypothesis assumes the priors
based on commonly available information could be set independently of each
other or other panels in the system. These are substantive assertions of course
but both $G_1$ and $G_2$ should be able to reflect on whether such assumptions
might be compelling. When considering the influence of their beliefs on each
other at least it is likely that $G_1$ and $G_2$ will be happy to accept the premises
of this theorem and so its conclusions, unless an experiment becomes
available that might confound these parameter vectors - see later.

 Note that this theorem holds irrespectively of the form of the utility function in the CK class: weaker conditions might guarantee adequacy for specific classes of utility factorizations \citep{Leonelli2015a}. Moreover, it applies whatever the definition of the underlying semi-graphoid, not only to probabilistic systems.  

\subsubsection{Likelihood separation in distributed probabilistic systems}
We now focus our attention onto  probabilistic systems and examine what soundness and distributivity might mean in this most common of contexts. Suppose an IDSS exhibits panel independence to $I_{+}^{0}$ at time $t=0$ so that the SB believes that $\independent _{i\in[m]}\boldsymbol{\theta }_{i}\ |\ (I_{+}^{0},d)$, $d\in\mathbb{D}$. Assume also that the only additional evidence presented to the IDSS by time $t$ by any panellist will be in the form of data sets $ \boldsymbol{x}^{t}=\left\{ \boldsymbol{x}_{\tau }:\tau \leq t\right\} $ which then populate $I_{+}^{t}$. The features that ensures the IDSS remains sound and distributed over time can be expressed in terms of the separability of a likelihood. Let $l(\boldsymbol{\theta }\ |\ \boldsymbol{x}^{t})$, $t\geq 0$, denote a likelihood over the parameter $\boldsymbol{\theta }$ of the distribution of $\boldsymbol{Y}$ given $\boldsymbol{x}^{t}$. Recall that subvectors of parameters associated with the probabilistic features delivered by the panel $G_{i}$ are denoted by $\boldsymbol{\theta }_{i},$ $i\in[m]$.

\begin{definition}
Call $l(\boldsymbol{\theta }\ |\ \boldsymbol{x}^{t})$ \emph{panel} \textbf{\emph{separable}} over the panel subvectors $\boldsymbol{\theta }_{i}$, $i\in[m]$, when, given admissible evidence $\boldsymbol{x}^{t}$, it is in the CK class  that for all $d\in \mathbb{D}$ 
\begin{equation*}
l(\boldsymbol{\theta }\ |\ \boldsymbol{x}^{t})=\prod_{i\in[m]}l_{i}(\boldsymbol{\theta }_{i}\ |\ \boldsymbol{t}_{i}(\boldsymbol{x}^{t}))
\end{equation*}
where $l_{i}(\boldsymbol{\theta }_{i}\ |\ \boldsymbol{t}_{i}(\boldsymbol{x}^{t}))$ is a function of $\boldsymbol{\theta }$ only through $\boldsymbol{\theta }_{i}$ and $\boldsymbol{t}_{i}(\boldsymbol{x}^{t})$ is a statistic of $\boldsymbol{x}^{t}$  known to $G_{i}$ and perhaps others, collected under the admissibility protocol and accommodated formally by $G_{i}$ into $I_{ii}^{t}$ to form its own posterior assessment of $\boldsymbol{\theta }_{i}$, $i\in[m]$.
\end{definition}

We now have the following theorem that gives good practical guidance about when and how soundness and delegatability can be preserved over time.

\begin{theorem}
Suppose an IDSS is adequate, delegable, separately informed, cutting and commonly separated at time $t=0$ and, for all times $t\geq 0$, data admitted to the system is panel separable at time $t$.  Then, provided the joint prior over $\boldsymbol{\theta }$ is absolutely continuous with respect to Lebesgue measure, the system is  sound and distributed at time $t$. On the other hand if at any time $t$ the system is not panel separable over a set of non-zero prior measures over the parameter vector $\boldsymbol{\theta }$ then the IDSS will no longer be sound or distributed. \label{theo:seplik}
\end{theorem}

See Appendix \ref{proof:seplik} for a proof of this result. So, for example, by designing  a single experiment to be orthogonal over parameters $\boldsymbol{\theta }_{i}$  and $\boldsymbol{\theta }_{j}$, where $\boldsymbol{\theta }_{i}$ are parameters in $G_i$'s model and $\boldsymbol{\theta }_{j}$ in $G_j$'s, ensures, under the conditions of the theorem, that an IDSS is sound and distributed, $i,j\in[m]$. From this \emph{single }experiment  data can still be included in the admissible evidence for both panels $G_i$ and $G_j$ and the system will still remain distributive.
Note, however, that the converse demonstrates that some protocol is certainly needed to preserve the distributivity property, even approximately.

Henceforth in this paper, until the discussion, we will assume all data admitted into an IDSS ensures this separability property. 

\subsection{Causality and the IDSS}
\subsubsection{Causal hypotheses: control and experiments} 
Recent advances have been made in formalising causal hypotheses in order to make inferences about the extent of a cause. Most of the original work in this area centred on BNs \citep{Pearl2000,Spirtes1993}.  However, the semantics have since been extended into other frameworks \citep[see e.g.][]{ Dawid2002, Dawid2010, Eichler2007, Lauritzen2002, Queen2009, Smith2007}. Typically, these assume that there is a partial order associated to the system \citep{Riccomagno2004} and that, within a causal framework and its implied partial order, the joint distributions of variables not upstream of a variable, which is externally controlled to take a particular value, remain unaffected by the control, whilst the effect on upstream variables is the same as if the controlled variable had simply taken that value. These causal semantics are a special case of the property described below in a sense we explain later. If this property is contained in the CK class then this greatly
simplifies the learning each panel needs to undertake.

\begin{definition}
Call an IDSS $\left( \mathbb{D},d^{0}\right)$-\emph{determined} if $\Pi _{i}^{\theta }= \left\{ \Pi _{i}^{\theta }(A_i,d):d\in \mathbb{D}, A_i\in\mathbb{A}_i\right\}$  is a stochastic function of $\Pi _{i}^{\theta }(A_i,d^{0})$, known to $G_{i}$, for some prescribed decision $d^{0}\in \mathbb{D}$ and a specific $A_i\in\mathbb{A}_i$, this being true for the beliefs of all panels $G_{i}$, $i\in[m]$.
\end{definition}

Once each panel has specified its beliefs about its own parameter vector $\boldsymbol{\theta} _{i}$ under a particular decision then, clearly, if the above property holds, the panel beliefs under other decisions can be calculated automatically. Obviously whether and how this condition is
satisfied depends heavily on the domain of the IDSS. However, surprisingly it is often met, sometimes implicitly as in the Kalman Filter \citep{Brockwell2002}. Perhaps more critically, it is also implicit for causal systems in the following sense.

\begin{lemma}
A causal BN, whose vertex  set is $\{Y_i: i\in[m]\}$, in a given CK class is a $\left( \mathbb{D},d^{0}\right)$-determined IDSS whenever the following three conditions hold: 
\begin{enumerate}
\item $\boldsymbol{\theta }$ consists of the entries in
the conditional probability tables of the BN; 
\item $d^{0}$ is the decision not to intervene but simply observe and $\mathbb{D}$ consists of decisions $d(y_{k})$ of setting a component $Y_{k}$ to one of its particular levels 
$y_{k}$; 
\item a single panel is responsible for the whole of a particular
conditional probability table of each $Y_{i}$, $i\in[m]$, given a particular
configuration of its parents.
\end{enumerate}
\end{lemma}

\begin{proof}
This derives directly from the definition in \citet{Pearl2008} of a causal BN.
Here we have three effects:
\begin{enumerate}
\item the impact of setting a variable $Y_{k}$ to a value $y_k$ will first have the effect of making the probability of the event $Y_k=y_k$ equal to one, conditional on any configuration of its parents. 
\item all conditional probability tables in the BN having $Y_k$ as a conditioning variable are replaced by the ones obtained by conditioning on the event $Y_k=y_k$.
\item the conditional probability tables of all other components of $Y$ are left unchanged. 
\end{enumerate}
This means, in particular that, as demanded by our definition, $\Pi _{i}^{\theta }(A_i,d(y_{k}))$ can be calculated as a simple function of $\Pi _{i}^{\theta }(A,d^{0})$. So this causal BN is $\left( \mathbb{D},d^{0}\right)$-determined.
\end{proof}

It also can be easily checked that, for example, causal hypotheses for other frameworks, such as chain event graphs \citep{Smith2008} and multiregression dynamic models (MDMs) \citep{Queen1993} provide a $\left(\mathbb{D},d^{0}\right) $-determined IDSS where $
\left( \mathbb{D},d^{0}\right) $ are analogously defined. Furthermore, at least in an adapted form, these causal hypotheses often relate to ones the SB would want to make within an IDSS. For example, in the scenario of exposure of cattle to disease, various culling regimes may be proposed to limit the exposure of susceptible cattle and help ameliorate the spread of the disease. Suppose
expert panel judgement has been based on what was observed to happen in a
parallel epidemic across farms when no controls were in place. If there were
common agreement that exposure, appropriately measured, really did `cause'
future infection, then a substantive, but plausible, hypothesis would be that the effects on the spread of the infection, if exposure were to be controlled by culling, could be identified with the effects when the \emph{effects} of this culling regime had occurred naturally. Probability predictions associated with this sort of control could then be inherited from the results of the observational study on the speed of disease spread when no control is exerted. The efficacy of various possible culling regimes could then be scored even if these regimes had never previously been enacted. Indeed this type of extrapolation is so common that the hypotheses on which it is based sometimes go unnoticed. Note that one particular consequence of this assumption is that different culling regimes giving rise to the same exposure profile, all other factors being equal, would be equally scored.

Within our context, it is helpful to extend the usual causal assumptions for three reasons. First, we may not need to require that the conditions hold for all possible levels of components of 
$\boldsymbol{Y}$,  but only a subset of these levels, broadly those that might improve, in some sense, the outcomes within the unfolding crisis. In practice only a very small proportion of the possible types of intervention considered in a casual model are usually entertained. It is, therefore, useful to acknowledge this less stringent assumption. Second, we might want the flexibility not to map the effects of enacting a control, but a rather more complex decision. This then embellishes $\mathbb{D}$. Last, the natural comparator $d^{0}$ for predicting what might happen may not be doing nothing but rather following routine procedures or past protocols. For all these reasons we have found very useful to generalise these definitions as above.

\subsubsection{Causal hypotheses and likelihood separation}

A second type of causal assumption is a useful addition to the CK class
because it enables panels to accommodate experimental as well as
observational data. Our motivation stems from work by \citet{Cooper1999}. They developed collections of additional assumptions which enabled formal learning of the parameters of discrete BNs when data was not only observational but could also come from designed experiments. They noted that if the BN was causal, in the sense given above where $\mathbb{D}$ contained the controls imposed on explanatory variables commonly used in setting up experiments, then experimental data could be introduced in a simple way. This technology has recently since been transferred to other domains \citep[e.g.][]{Freeman2011, Cowell2014}.

Different panels in most IDSSs will want to accommodate not only observational and survey data but also experimental data. Here covariates are controlled and set to certain values and the subsequent effect on a response variable observed.  So, to be able to accommodate relevant scientific evidence, most operational IDSSs will need to assume causal hypotheses concerning controlled experiments overseen by a relevant panel as this applies to an unfolding crisis. Interestingly, it has been shown \citep[see][]{Daneshkhah2004} at least within the context of BNs that the panel independence assumption necessary to ensure distributivity of an IDSS is intimately linked with, and plausible only when, certain causal hypotheses can be entertained. These causal hypotheses concern not only experiments \emph{within} a module used by a panel but also \emph{across the interface} of the modules. Since the first class of hypotheses is typically a function of a single decision support system in this paper we will concentrate mainly on the second case.

\begin{definition}
Call an IDSS  $\boldsymbol{e}^{t}-$\emph{panel compatible} for a collection of experiments $\boldsymbol{e}^{t}\boldsymbol{=}\left(e_{i}^{t}\right)_{i\in[p]} $ leading to a data set $\boldsymbol{x}^{t}$ if the likelihood associated with $\bm{e}^t$ can be written as 
\begin{equation*}
l^{e}(\boldsymbol{\theta }\ |\ \boldsymbol{t}(\boldsymbol{x}^{t}))=\prod\limits_{i\in[m]}l_{i}^{e}(\boldsymbol{\theta }_{i}\ |\ \boldsymbol{t}_{i}(\boldsymbol{x}^{t}))
\end{equation*}
where $l_{i}^{e}(\boldsymbol{\theta }_{i}\ |\ \boldsymbol{t}_{i}(\boldsymbol{x}^{t}))$ is a function only of the parameters $\boldsymbol{\theta }_{i}$ overseen by a single panel $G_{i}$, $i\in[m]$.
\end{definition}

By far the most common such collection of experiments is one composed of collections of independent experiments that can be partitioned so that each experiment is informative only about the parameters overseen by each particular panel. Again an important special case when this separation will be automatic is when the overarching qualitative structure is a BN and this BN is causal under experimental manipulations. 

\begin{lemma}
Suppose an IDSS is $\left( \mathbb{D},d^{0}\right)$-determined where $d^{0}\in \mathbb{D}$ consists of no intervention on the system. Then the IDSS is $\boldsymbol{e}$-panel compatible whenever $\bm{e}^t$ is made up of components $e_k^t$ consisting of independent, double blind, randomised, designed experiments where the response random vector $\boldsymbol{Y}_{i}$ denotes the observed number of units  in a causal BN when its parents take their particular possible vector of  configuration of values, $i\in[m]$, $k\in[p]$.
\end{lemma}

\begin{proof}
One of the properties of a causal BN is that the probability of a random sample of a set of observations of one of its nodes when its parents are manipulated to take a certain value can be equated with the probability of observing this sample when it is idle - i.e. when there is no intervention. In particular this means that the $\boldsymbol{\theta }_{i}$ appearing in
the manipulated experiment is equal to $G_{i}$'s corresponding parameters of interest in the observational setting where $d^{0}\in \mathbb{D}$, $i\in[m]$. But because this is so for the idle control $d^{0}$, since the IDSS is $\left( \mathbb{D},d^{0}\right) $-determined, it is also true for all $d\in \mathbb{D}$.
\end{proof}

Note that as in the remarks of the last section, these are not the only class of experiments to have this property but simply an important special case. The directly analogous definitions of causal hypotheses as they apply to different overarching graphical structures like chain event graphs or MDMs also admit such a factorisation.

\section{Examples of sound and distributive frameworks}
\label{sec:frameworks}
We saw in Section \ref{sec:coherence} that, provided a system is such that certain qualitative properties existed over the parameters of a composite model and the likelihood of information separated over parameters in an appropriate way, then the composite system should remain distributed with all the implementation, interpretative and computational advantages that confers. But how common are such systems? The answer is that whilst many systems violate the conditions needed, many others satisfy them. So by choosing panels appropriately and by demanding that only certain types of unambiguously interpretable data are entered into the IDSS, it is often possible to build such distributed systems. In this section we present some well known settings where this is possible and one where it is not.

\subsection{Stochastic and value independence}
We begin with a trivial system made up of independent components and linear utility.

\begin{definition}
A centre's reward vector  $\bm{Y}=(\bm{Y}_i)_{i\in[m]}$ is said to have \textbf{value independent attributes}, and the utility has a \textbf{linear form}, if $U(d,\bm{y})$ can be written as
\begin{equation}
U(d, \bm{y})=\sum_{i\in[m]}k_{i}(d)U_{i}(d,\bm{y}_{i} )  \label{value ind},
\end{equation}
where $\sum_{i\in[m]}k_{i}(d)=1$, $k_i(d)\in(0,1)$ are called \emph{criterion weights} and $\sup_{d,\bm{y}_{i}}\;U_{i}(\bm{y}_{i},d )=1$, $i\in [m]$, $d\in\mathbb{D}$.
\end{definition}

Suppose the CK-class contains the hypothesis that all potential users will have value independent attributes. Assume further that the vectors  $\bm{Y}_{i}$  are mutually independent for any given $d\in \mathbb{D}$ so that $\independent_{i\in[m]}\bm{Y}_{i}\;|\;\bm{\theta }_{i},d$, where the distribution of $\bm{Y}_{i}$ is a known deterministic function of $\bm{\theta }_{i}$ and $d\in\mathbb{D}$, $i\in[m].$ For each $d\in \mathbb{D}$, let the expected utility $\bar{U}_{i}(\bm{\theta }_{i}, d)$ denote the expectation of $U_{i}(\bm{y}_i,d )$ over $\bm{Y}_{i}\;|\;\bm{\theta }_{i},d$, and $\bar{U}(d)$ the expectation over $\bm{Y}_i\;|\;d$, $i\in[m]$. Then 
\begin{equation*}
\bar{U}(d)=\sum_{i\in[m]}k_{i}(d)\bar{U}_{i}(d),  
\end{equation*}
where
\begin{equation*}
\bar{U}_{i}(d)=\int_{\bm{\theta }_{i}\in \Theta _{i}(d)}\bar{U}_{i}(\bm{\theta}_i,d)\pi _{i}(\bm{\theta }_{i}\;|\;d)\dr\bm{\theta }_{i},
\end{equation*}
and $\pi _{i}(\bm{\theta }_{i}\;|\;d)$ is SB's /$G_{i}$'s prior on $\bm{\theta }_{i}\;|\;d$, $d\in\mathbb{D}$. This system is clearly \emph{a priori} distributive. So the SB can devolve her calculations of the expected utility of each $d\in \mathbb{D}$ to the relevant panels. She can then use these expected utility evaluations with a user's input $\bm{k}(d)=(k_i(d))_{i\in[m]}$ to evaluate $\bar{U}(d)$ and hence discover an optimal policy \citep[as in][]{Dodd2012,Smith2012}.

Note that in this scenario no panel needs to contemplate any issue other than those concerning its own domain of expertise. The IDSS can simply re-evaluate scores if changing circumstances demand that the user's priorities, reflected in her choice of criterion weights, need to be modified, using the original delivered information.  However, as soon as the utility is not linear, or there is a dependence induced across the components $\bm{Y}_i$,  $i\in[m]$, then this need no longer to be the case \citep{Leonelli2015}. 

\subsection{Staged trees}
\label{sec:tree} 
Perhaps the simplest overarching structure within which to express forms of dependence is the event tree $\mathcal{T}(d)$, for each $d\in \mathbb{D}$. Suppose  that panels can agree the topology of the underlying event tree  $\mathcal{T}(d)$ for each possible decision describing the set of corresponding unfolding events.  Suppose further that all agree that panel $G_{i}$, $i\in[m]$, should deliver the edge probability vectors $\bm{\theta }_{ij}$ associated with edges emanating from each non-leaf vertex $v_{ij}$, $j\in[m_i]$. Let $\bm{\theta }_{ij}\triangleq \left( \theta _{ijk}\right)_{k\in[m_{ij}]} $, $\bm{\theta }_{i}\triangleq \left( \bm{\theta }_{ij}\right)_{j\in[m_i]} $ and $\bm{\theta }\triangleq \left( \bm{\theta }_{i}\right)_{i\in[m]} $. Then, for a random sample $\bm{x}^t$,
\begin{align*}
l(\bm{\theta }\;|\;\bm{x}^{t})&=\prod_{i\in[m]}l_{i}(\bm{\theta }_{i}\;|\;\bm{t}_{i}(\bm{x}^{t})),\\
\shortintertext{where} 
l_{i}(\bm{\theta }_{i}\;|\;\bm{t}_{i}(\bm{x}^{t}))&=\prod_{j\in[m_i]}l_{ij}(\bm{\theta }_{ij}\;|\;\bm{t}_{i}(\bm{x}^{t})),\\
\shortintertext{and} 
l_{ij}(\bm{\theta }_{i}\;|\;\bm{t}_{i}(\bm{x}^{t}))&=\prod_{k\in[m_{ij}]}\theta _{ijk}^{x_{ijk}},
\end{align*}
where $\sum_{k\in[m_{ij}]}\theta _{ijk}=1$ and $x_{ikj}$ is the number of units in the sample reaching vertex $v_{ij}$ and then proceeding down the $k^{th}$ edge, $i\in[m]$, $j\in[m_i]$ and  $k\in[m_{ij}]$. Here $\bm{t}_i=\{x_{ijk}:j\in[m_i], k\in[m_{ij}]\}$. Clearly, in many cases, including complete datasets, this likelihood separates over the panels.  It can also be shown that if the event tree is causal, as defined in \citet{Cowell2014}, \citet{Shafer1996} and \citet{Thwaites2010}, then the likelihood remains separable. 

An example of this sort of structure comes from level activity modelling in forensic science \citep{Cowell2011, Dawid2007}. When collecting evidence concerning a particular criminal prosecution, inference can sometimes be expressed in terms of such a tree where the different panels will be the jury members, the forensic experts, and the court recorders: see \citet{Smith2010} for an example in this context. Then under panel independence the whole system will enjoy the property of being distributed. Note that in this sort of application, when the prosecution and defence case have a different narrative about an activity, e.g. the defence case asserts the suspect was not at the scene of the crime whilst the prosecution case asserts that the suspect went through a sequence of actions, then the corresponding event tree is very asymmetric.  In these sorts of examples, building the IDSS on a BN would be inefficient and contrived. This is one reason why our theoretical development is not embedded in a BN framework, but instead develops a methodology which extends to BNs as a special case.

\subsection{Bayesian networks and chain graphs}
Currently, the most well developed and established qualitative structure is the BN. In our context the different panels would be asked to agree to a particular directed acyclic graph (DAG) of the BN. The responsibility for the delivery of different conditional probability tables (or more generally rows in these tables) would be partitioned out across the different panels. In this setting, global independence and local independence are sufficient for panel independence.  In practice, these assumptions are almost always made in the more usual one-panel setting \citep{Spiegelhalter1990}. These make probability elicitations and computations much more straightforward than they would otherwise be. The separation of the likelihood for both discrete and continuous BNs under complete or ancestral sampling has long been recognised \citep[e.g.][]{Lauritzen1996,Smith2010}, as has the closure to designed experiments \citep[e.g.][]{Cooper1999}.

Let $\Gr$ be a BN with vertex set $\{Y_i:i\in[n]\}$ and $\{B_i: i\in[m]\}$ be a partition of $[n]$. Suppose panel $G_i$ oversees $\bm{Y}_{B_i}=(Y_j)_{j\in B_i}$, $i\in[m]$. In this setting panel independence implies 
\[
\begin{array}{ccc}
\pi(\bm{\theta}\;|\;d)=\prod_{i\in[m]}\pi_i(\bm{\theta}_{B_i}\;|\;d), & \mbox{ and } &
f(\bm{y})=\prod_{i\in[m]}f_i\left(\bm{y}_{B_i}\;|\;\bm{y}_{\Pi_{B_i}}\right),
\end{array}
\]
where $\Pi_{B_i}=\cup_{j\in B_i}\Pi_j$, $\Pi_j$ is the index parent set of $Y_j$, $\bm{\theta}_{B_i}=(\bm{\theta}_j)_{j\in B_i}$, $\bm{\theta}_j$ being the parameter associated to a vertex $Y_j$, and 
\begin{equation*}
f_i\left(\bm{y}_{B_i}\;|\;\bm{y}_{\Pi_{B_i}}\right)=\int_{\bm{\theta}_{B_i}\Theta_{B_i}(d)}f_i(\bm{y}_{B_i}\;|\;\bm{y}_{\Pi_{B_i}},\bm{\theta}_{B_i})\pi_i(\bm{\theta}_{B_i}\;|\;d)\dr \bm{\theta}_{B_i}, \label{eq:bnidss3}
\end{equation*}
where $\Theta_{B_i}(d)$ is the parameter space of panel $G_i$, $d\in\mathbb{D}$.

Furthermore, provided that $\bm{x}^t$ is a complete random sample from the same population as $\bm{Y}$, the posterior density can be written as
\[
\pi(\bm{\theta}\;|\;\bm{x}^t,d)=\prod_{i\in[m]}\pi_i(\bm{\theta}_{B_i}\;|\;\bm{x}^t_{B_i},\bm{x}^t_{\Pi_{B_i}},d).
\]
 For object-oriented BNs which construct a large, overarching BNs through copying parts of its specification \citep{Koller1997} we would need the additional condition that all replicates of a probabilistic `object' were the responsibility of a single panel. But since such replicates are asserting that the parameters of certain aspects of the model are identical this condition would, in practice, normally be satisfied.  For a chain graph \citep{Wermuth1990, Frydenberg1990,Lauritzen2002}, provided that all the probability specifications concerning variables in a given chain component are always the responsibility of a single panel, the assumption of prior panel independence can be made consistently and is natural. Again under complete or ancestral sampling the likelihood also separates over the panels, which was exploited for the combination of expert judgement in \citet{Faria1997}.

Although some common BNs are often coded as discrete or Gaussian and supported by a conjugate analysis, this is not a requirement for their definition via the semi-graphoid axioms (stated on page \pageref{Def:SemiGraphoid}). For example any panel's distribution can be extremely complex, involving many latent variables which are later marginalised out to provide the required, sometimes non-parametric, posterior distributions over many or even an infinity of parameters. Usually then posterior distributions will not have an explicit algebraic form, but constitute large sample approximations from the complex system. In our context it would be extremely cumbersome, even in small scale  problems, for a panel to communicate such information in its entirety to the system. However, in a rather different context of addressing the challenges of inference using massive data sets, \citet{Bissiri2013} have already noted  that only certain moment are needed for the calculations of expected utilities for certain classes of utilities. This property can be exploited when considering an overarching framework of a BN when one component of that system is very large \citep [see below and e.g.][]  {Lauritzen1992, Nilsson2001}. This is also the case for an IDSS. It is interesting to note that, in fact, propagation algorithms for general BNs in terms of lower order moments and general distributional assumptions have been known for a long time. 

\subsection{Decomposable graphs and cliques of panellists}
One class of problems where the sorts of conditions needed for separation can be fierce is when the overarching model is an undirected graph \citep{Lauritzen1996}. From \citet{Dawid1993} we know that a sensible way to build panels would be by letting each oversee a clique of an agreed decomposable undirected graph, since distributed inferences can then be carried out for complete observations. However, now two different panels overseeing cliques adjacent to one another will \emph{both} have responsibility for parameters of the margins of variables lying on the separator. Therefore variational independence does not hold, panel independence is broken and  there is a danger that the system is not well defined. Two panels need to agree the same prior over this separator, for example, ensured through imposing into the CK-class that the centre's beliefs are strong hyper-Markov \citep{Dawid1993}. Even then, if the two adjacent panels update their beliefs autonomously with different sets of information, then there is no guarantee that the resulting two posterior distributions of the two different panels will remain hyper-Markov.

One simple solution to this problem, when the graph is decomposable, is to give precedence to one panel's information about a particular separator and ignore all others'. This is equivalent to selecting a BN representation of this decomposable graph where responsibility for a particular separator is delegated entirely to the panel delivering the parent clique probabilities. The other panel is then only responsible for delivering the probability judgements about its clique probabilities conditional on the values of the separator: see related issues associated with meta-analyses of incomplete data \citep[e.g.][]{Massa2010, Jirousek2003}. In these sorts of contexts the standard Bayesian paradigm may well not be ideal. More expressive inference, perhaps using belief functions \citep{Shafer1976}, might be more appropriate to represent the composition of beliefs in such settings. 

\subsection{MDMs and uncoupled dynamic BNs}
BN time series models, such as the general dynamic BN \citep{Koller2001,Murphy2002} or the dynamic matrix-variate graphical model \citep{Carvalho2007}, useful for modelling different components of the system, do not exhibit the required distributive properties, so these do not provide good frameworks for the overarching system. However, the class of MDM models \citep{Queen1993} does \citep[see, for application of this class,][]{Anacleto2013b, Costa2015, Queen1994}. This essentially models a multivariate time series $\{\bm{Y}(t)\}_{t\in\mathbb{Z}_{\geq 1}}=\{Y_1i(t): i\in[n]\}_{t\in\mathbb{Z}_{\geq 1}}$ as a DAG whose vertices are univariate processes and its topology remains fixed through time. The univariate time series are modelled with a regression dynamic linear model \citep{West1997}, where the regressors are the contemporaneous time series. Importantly both the predictive, $f(\bm{y}(t)\;|\;\bm{y}^{t-1})$, $\bm{y}^{t-1}=(\bm{y}(s))_{s\in[t-1]}$, and conditional, $f(\bm{y}(t)\;|\;\bm{\theta}(t),\bm{y}^{t-1})$, distributions of an MDM factorize accordingly to the underlying DAG.  This ensures distributivity.
Specifically, for $i\in[n]$ we assume:
\begin{align}
& Y_i(t)=\bm{\theta}_i(t)\bm{F}_i(t)^\T+\bm{v}_i(t),				& & v_i(t)\sim (0,V_i(t)) ; \label{Eq:MDM1}\\
& \bm{\theta}(t)^\T=G(t)\bm{\theta}(t{-1})^\T+\bm{w}(t),		& &\bm{w}(t)\sim (\bm{0},W(t));\label{Eq:MDM2}\\
& \bm{\theta}(1)\;|\;I^0 \sim (\bm{m}(1),C(1)). 		 				&	& \label{Eq:MDM3}
\end{align}
Here $\bm{F}_t(i)\in\mathbb{R}^{s_i}$, $s_i\in\mathbb{Z}_{\geq 1}$, is a vector of regressors, known functions of $\bm{Y}^t_{\Pi_i}$ and $\bm{Y}^{t-1}_i$. The  parameter vector $\bm{\theta}(t)=(\bm{\theta}_i(t))_{i\in[n]}\in\mathbb{R}^s$ has elements $\bm{\theta}_i(t)\in\mathbb{R}^{s_i}$, where $\sum_{i\in[n]}s_i=s$. The observation error $v_i(t)$, that can be either assumed known or unknown, is univariate and $V_i(t)\in\mathbb{R}_{>0}$ is the observational variance, often assumed constant through time. The  known matrices  $G(t)=\textnormal{blockdiag}(G_1(t),\dots,G_n(t))$ and $W(t)=\textnormal{blockdiag}(W_1(t),\dots, W_n(t))$ have dimension $s\times s$, where $G_i(t)$ and $W_i(t)$ are  the $s_i\times s_i$ state evolution matrix and state evolution covariance matrix for $\bm{\theta}_i(t)$, respectively, $i\in[n]$. The $s$-dimensional vector $\bm{w}(t)=(\bm{w}_i(t))_{i\in[n]}$ has elements $\bm{w}_i(t)$ of dimension $s_i$ called system error vector. The errors $v_1(t),\dots, v_n(t),\bm{w}_1(t),\dots, \bm{w}_n(t)$ are assumed independent of each other and through time. The vector $\bm{m}(1)\in\mathbb{R}^s$ and the $s\times s$  covariance matrix $C(1)=\textnormal{blockdiag}(C_1(1),\dots,C_n(1))$ are the first two moments of the distribution of $\bm{\theta}(1)\;|\; I^0$, where $I^0$ represents the initial information available.


To illustrate the features of the MDM model, consider the DAG of a na\"{i}ve food supply model in Figure \ref{fig:foodMDM} and suppose panel $G_i$ oversees the time series $\{Y_i(t)\}_{y\in\mathbb{Z}_{\geq 1}}$ for each decision $d\in\mathbb{D}$. It can be shown that the joint distribution of $(\bm{Y}_t,\bm{\theta}_t)$ given the past factorizes as
\begin{equation*}
f(\bm{y}(t),\bm{\theta}(t)\;|\; I^{t-1},d)=\prod_{i\in[4]}f_i(y_i(t)\;|\;\bm{y}_{\Pi_i}(t),\bm{\theta}_i(t), I^{t-1})\pi_i(\bm{\theta}_i(t)\;|\;I^{t-1},d).
\end{equation*} 
So note the critical feature that the panel parameters are independent of each other at every time point conditional on the past. As mentioned, the predictive distribution also factorizes accordingly as

 \begin{equation}
 \label{eq:mdmmarg}
 f(\bm{y}_t\;|\;\bm{y}^{t-1})=\prod_{i\in[4]} g_{t,i}(\bm{y}^t_i,\bm{y}^t_{\Pi_i},\bm{\theta}_i(t),d),
 \end{equation}
where
\begin{equation*}
g_{t,i}=\int_{\Theta_i(d)}f_i(y_i(t)\;|\;\bm{y}^t_{\Pi_i},\bm{y}^{t-1}_i,\bm{\theta}_i(t))\pi_i(\bm{\theta}_i(t)\;|\;\bm{y}^{t-1}_{\Pi_i},\bm{y}_i^{t-1},d)\textnormal{d} \bm{\theta}_i(t).
\end{equation*}

Although the predictive joint distribution of $\bm{Y}(t) $ is not Gaussian, indeed its distribution can be  very complex, the joint moments of its predictives are straightforward to calculate and are polynomial functions of various posterior moments that can be autonomously calculated by the different panel members \citep[see][for an explicit derivation of some of these recurrences]{Queen2008}. As discussed in \citet{Leonelli2015}, in this framework, expected utilities are often functions of these moments. Therefore, panels need only to deliver to the system a few summaries, making the calculation of expected utility scores very fast.

Note that for modularity to hold we do not need any distributional assumptions other than panel independence, but simply, as for the BN, that the individual component processes can be partitioned so that exactly one panel is responsible for the development of that process, given its parents. This is because distributivity derives only from the assumed conditional independence structure and not from the distributional assumptions about the parameters or conditional distribution of responses given parents and these parameters. So instead of components being dynamic regression models, they can essentially have an arbitrary structure, provided they give an explicit distribution for the response variable given the parameters and the configuration of the parents. 

\section{Panel summary recursions for a dynamic IDSS}
\label{sec:MDM}
\label{sec:dynamic}

\subsection{An instance of a dynamic CK-class}

\label{StructuralAssumptions}

In \citet{Leonelli2015} and \citet{LeonelliThesis} five general structural properties were discovered which arise as a consequence of the axioms of Section \ref{sec:toy}. We were able to show that these properties ensured that the expected utilities needed by the IDSS to score options could be calculated using various tower rule expansions. Critically, we showed that each term in these expansions, usually a polynomial, could be calculated autonomously by a single panel, once it had communicated with its neighbours, so the required expectations then could be calculated through appropriate message-passing between constituent panels. Many of the classes of models we discuss there satisfy the conditions needed in this paper.  The general forms of the tower rule formulae given in \citet{Leonelli2015} and \citet{LeonelliThesis} needed to identify expected utility maximising decisions, although straightforward, are  very general and rather technical.  So here we restrict ourselves to illustrating how these message-passing algorithms might work for the elementary system given in Figure \ref{fig:foodMDM}.
\begin{figure}
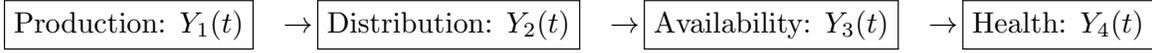

\begin{equation*}
\resizebox{.95\hsize}{!}{$\begin{array}{ccccccc}
\fbox{Production: $Y_1(t)$}  \hspace{-.3cm}& \rightarrow & \hspace{-.3cm}\fbox{Distribution: $Y_2(t)$}\hspace{-.3cm}&\rightarrow
 &\hspace{-.3cm}\fbox{Availability: $Y_3(t)$}\hspace{-.3cm} &\rightarrow &\hspace{-.3cm}\fbox{Health: $Y_4(t)$} 
\end{array}$}
\end{equation*}
\caption{A na{\"i}ve DAG for food security. \label{fig:foodMDM}}
\end{figure}

Thus suppose the space of possible decisions, $\mathbb{D}$, is sufficiently  small for it to be plausible to evaluate the expected utility of every policy separately, one at a time. Without loss, for simplicity we drop the $d \in \mathbb{D}$ index, and present the formulae needed for the expected utility of a fixed decision to be calculated. Let $\bm{Y}^{\T}=(\bm{Y}(1)^{\T},\dots,\bm{Y}(T)^{\T})$ be a random vector observed over a finite time horizon $T$ where $\bm{Y}(t)$ is the vector at time $t\in[T]$. Let $\bm{Y}(t)^{\T}=(\bm{Y}_1^{\T},\dots, \bm{Y}_4^{\T})$, where $\bm{Y}_i(t)$ is under the responsibility of panel $G_i$, and $\bm{Y}_i(t)^{\T}=(Y_i^1(t),\dots, Y_i^r(t))$ be a vector of variables observed over $r$ different locations in space, $i\in[4]$, $t\in[T]$. Suppose that the family $\mathbb{U}$ in the CK class demands that the delivered utility function is linear over time, space and panel index. It follows that  $U$ must take the form:
\[
U=\sum_{i\in[4]}\sum_{t\in[T]}\sum_{l\in[r]}k_{til}U_{til}(y_{i}^l(t)),
\]
Assume further that the marginal utilities  $U_{til}=-\gamma_i(t)y_i^l(t)^2$ where $\gamma_i(t) \in \mathbb{R}$. Here the criterion weights $k_{til}$ are fixed constants, having been previously elicited.  

Assume panel $G_1$ uses any dynamic probabilistic model of arbitrary complexity, able to deliver as output the first two moments of the distribution of $\bm{Y}_1(t)$, $t\in[T]$. Let $\bm{a}_1(t)$ and let $C_1(t)$ denote the mean and covariance matrix, respectively, of $\bm{Y}_1(t)$, where
\[ C_1(t)=\left[ \begin{array}{c}
c_1^{p,s}(t) \\
\end{array} \right]_{p,s }.\] 
Also let $a_1^l(t)$ and $c_1^l(t)$ denote the mean and the variance of $Y_1^l(t)$ and $\bm{c}_1(t)=(c_1^1(t),\dots,c_1^r(t))$.

The second panel, $G_2$, uses as its model of $\bm{Y}_2(t)$ $r$ different MDMs. Equation (\ref{Eq:MDM1}) then becomes: 
\begin{align*}
&Y_2^l(t)=\bm{Y}_1(t)^\T\bm{\theta}_2^l(t) +v_2^l(t),\;\;\;\;\;\;\;\;v_2^l(t)\sim (0,V^l_2(t)),
\end{align*}
where $\bm{\theta}_2^l(t)$ is an $r$-dimensional vector of unknown parameters and $\E(V_2^l(t))=b_2^l(t)$. The system equations (\ref{Eq:MDM2}) therefore become:
\[
\bm{\theta}_2^l(t)=G_{2}^l(t)\bm{\theta}_{2}^l(t{-1})+\bm{w}_{2}^l(t),\;\;\;\;\;\;\;\;\bm{w}_{2}^l(t)\sim (\bm{0},W_{2}^l(t)),
\]
where $G_{2}^l(t)$ is a $r\times r$ matrix having known entries and   
\[W_{2}^l(t)= \left[ \begin{array}{c}
w_{2,l}^{i,j} \\
\end{array} \right]_{i,j }\] 
 is a matrix. Here the panel has ascertained $\bm{\theta}^l_2(1)$ has a prior distribution with mean $\bm{a}_2^{l}$ and covariance $C_{2}^l$, where 
\[ C_{2}^l=\left[ \begin{array}{c}
c_{2,l}^{i,j} \\
\end{array} \right]_{i,j .}\]

Panel $G_3$ also employs a simple linear MDM defined, for $l\in[r]$, by the following observation and system equations 
\[
Y_3^l(t)=\theta_{23}^l(t)Y_2^l(t)+v_3^l(t), \,\,\,\,
\theta_{23}^l(t)=\theta_{23}^l(t{-1})+w_{23}^l(t).
\]
The errors $v_t^l(3)$ and $w_t^l(2,3)$ are assumed by $G_3$ to be mutually independent with mean zero and variance $V_3^l(t)$ and $W_{23}^l(t)$, respectively. These variances are assumed to be unknown, but the panel provides a prior mean $b_3^l(t)$ for $V_3^l(t)$ and $r_{23}^l(t)$ for $W_{23}^l(t)$. Prior means and variances are delivered also for the parameters $\theta_{23}^l(1)$ and denoted as $a_{23}^l$ and $c_{23}^l$.

The final panel, $G_4$, uses an emulator based on a complex deterministic simulator for each region; for our example, the household demography around a given region.  This model might be a simulator of purchasing activities over the space of households, with runs covering a small sample from retail outlets in the region. The inputs of each simulator are $y_3^l(t)$ together with other known constants, whilst its output is $y_4^l(t)$. The emulator is then formally defined as
\[
Y_4^l(t)=m(Y_3^l(t))+e(Y_3^l(t)),
\]
where $m(Y_3^l(t))=\theta^l_{04}+\theta^l_{34}Y_3^l(t)$ and $e(Y_3^l(t))$ is a zero mean Gaussian process with covariance function $c^l(\cdot,\cdot)=W^l_{4}r^l(Y_3^l(t)-Y_3^l(t)^{'})$, where $r^l$ is a stationary correlation function such that $r^l(0)=1$ agreed by panel $G_4$ using their expert judgements \citep[see e.g.][]{Kennedy2001} for this emulator. From the emulator $G_4$ is now able to calculate, in particular, the moments of his parameters which will be needed by the SB as
\[
\begin{array}{lllll}
\mathbb{E}(\theta^l_{04})=a^l_{04},&\mathbb{E}(\theta^l_{34})=a^l_{34},&
\mathbb{V}(\theta^l_{04})=c^l_{04},&\mathbb{V}(\theta^l_{34})=c^l_{34},&\mathbb{E}(W^l_{4})=r^l_{4}.
\end{array}
\]

The algorithm, described in \citet{Leonelli2015}, enables us to customise a message-passing algorithm for the overarching structure defined.  This works backwards both through the DAG and over time, from the last panellist at the last time point, and entails the computation of the function $\tilde{u}_{t,j}$, which includes all the terms that are a function of $\bm{Y}_j(t)$ in the backward routine, and its expectation $\bar{u}_{t,j}$.  We will illustrate the algorithm over just two time steps, so start from panel $G_4$ which oversees $\bm{Y}_4(2)$. Here we set $\tilde{u}_{2,4}=U$, and let $\dot{u}_{2,4}=U-\sum_{l\in[r]}k_{24l}U_{24l}$. Panel $G_4$ first computes $\bar{u}_{2,4}$, using the formula
\begin{align*}
\bar{u}_{2,4}&=\dot{u}_{2,4}-\sum_{l\in[r]}k_{24l}\gamma_4^l(2)\mathbb{E}(Y_4^l(2)^2)=\dot{u}_{2,4}-\sum_{l\in[r]}k_{24l}\gamma_4^l(2)(\mathbb{E}(Y_4^l(2))^2+\mathbb{V}(Y_4^l(2)))\\
&=\dot{u}_{2,4}-\sum_{l\in[r]}\gamma_4^l(2)k_{24l}\left[\tau^l_{4}+2a^l_{04}a^l_{34}\mathbb{E}(Y_3^l(2))+\tau^l_{34}\mathbb{E}(Y_3^l(2)^2)\right],
\end{align*}
where $\tau^l_{04}={a^l_{04}}^2+r^l_{4}+c^l_{04}$ and $\tau^l_{34}={a^l_{34}}^2+c^l_{34}$.

This function is then  communicated to $G_3$. Because of the topology of the DAG in Figure \ref{fig:foodMDM}, under the notation in \citet{Leonelli2015}, $\tilde{u}_{2,3}=\bar{u}_{2,4}+\sum_{l\in[r]}U_{23l}$. So writing $h_{2,3}^l=-2\gamma_4^l(2)k_{24l}a_{04}^la_{34}^l$ and 
\[
\dot{u}_{2,3}=\sum_{i\in[4]}\sum_{l\in[r]}k_{1il}U_{1il}+k_{21l}U_{21l}+k_{22l}U_{22l}-\gamma_4^l(2)k_{24l}\tau_{4}^l,
\]  
it follows that $\bar{u}_{2,3}$ is then equal to 
\begin{align*}
\bar{u}_{2,3}&=\dot{u}_{2,3}-\sum_{l\in[r]}(\gamma_4^l(2)k_{24l}\tau_{34}^l+k_{23l}\gamma_{3}^l(2))\E(Y_{3}^l(2)^2)+h_{2,3}^l\E(Y_3^l(2)),
\end{align*}
where $\E(Y_3^l(2))=a_{23}^l\E(Y_2^l(2))$ and $\E(Y_3^l(2)^2)=({a^l_{23}}^2+r_{23}^l(2)+c_{23}^l)\E(Y_2^l(2)^2)+b_3^l(2).$  Noting that $\tilde{u}_{2,2}=\bar{u}_{2,3}+\sum_{l\in[r]}U_{22l}$, rearranging, we obtain the equation
\begin{equation}
\label{eq:kmn}
\tilde{u}_{2,2}=\dot{u}_{2,2}-\sum_{l\in[r]}\left[[(\gamma_4^l(2)k_{24l}\tau_{34}^l+k_{23l}\gamma_{3}^l(2))({a_{23}^l}^2+r_{23}^l(2)+c_{23}^l)+k_{22l}\gamma_{2}^l(2)]\E(Y_2^l(2)^2)+h_{2,3}^la_{23}^l\E(Y_2^l(2))\right],
\end{equation}
where 
\begin{equation}
\label{eq:kmnp}
\dot{u}_{2,2}=\sum_{i\in[4]}\sum_{l\in[r]}k_{1il}U_{1il}+k_{21l}U_{21l}-\gamma_4^l(2)k_{24l}\tau_{4}^l-(\gamma_4^l(2)k_{24l}\tau_{34}^l+k_{23l}\gamma_{3}^l(2))b_3^l(2).
\end{equation}

At this stage panel $G_2$ performs a marginalisation step computing 
\begin{equation}
\E(Y_2^l(2))=\E(\bm{Y}_1(2)^\T)G_2^l(2)\bm{a}_2^l,\,\,\mbox{and}\,\,
\E(Y_2^l(2)^2)=\E(Y_2^l(2))^2+\V(Y_2^l(2)),
\end{equation}
where
\begin{equation}
\E(Y_2^l(2))^2={\bm{a}_2^l}^\T G_2^l(2)^\T \E(\bm{Y}_1(2))\E(\bm{Y}_1(2)^\T)G_2^l(2){\bm{a}_2^l}
\end{equation} 
and  
\begin{align}
\V(Y_2^l(2))&=\V(Y_1(2)^\T\bm{\theta}_{2}^l(2))+\E(V_2^l(2))= \V(\bm{Y}_1(t)^\T G_2^l(2)\bm{\theta}_2^l(1))+\E(\bm{Y}_1^\T W_2^l(2)\bm{Y}_1(t))+b_2^l(2)\nonumber\\
&=\V(\bm{Y}_1(t)^\T G_2^l(2)\bm{a}_2^l)+\E(\bm{Y}_1(2)^\T G_{2}^l(2)C_2^lG_2^l(2)^\T \bm{Y}_1(t))+\E(\bm{Y}_1^\T W_2^l(2)\bm{Y}_1(t))+b_2^l(2)\nonumber\\
&= {\bm{a}_2^l}^\T G_2^l(2)^\T \V(\bm{Y}_1(2))G_2^l(2)\bm{a}_2^l+\E(\bm{Y}_1(2)^\T (G_{2}^l(2)C_2^lG_2^l(2)^\T+W_2^l(2)) \bm{Y}_1(2))
+b_2^l(2)\label{eq:kmnfgs}.
\end{align}

Substituting the results of  equations (\ref{eq:kmnp})-(\ref{eq:kmnfgs}) into (\ref{eq:kmn}) panel $G_2$ can now compute  $\bar{u}_{2,2}$. This function is then sent to $G_1$. Recalling that $\E(\bm{Y}_1(2))=\bm{a}_1(2)$ and $\V(\bm{Y}_1(2))=C_1(2)$, we see that 
\[
\E(\bm{Y}_1(2)^\T (G_{2}^l(2)C_2^lG_2^l(2)^\T+W_2^l(2)) \bm{Y}_1(2))= \sum_{\mathclap{i,j\in[r]}}(a^{i,j}+w^{i,j}_{2,l})(a^i_1(2)a_1^j(2)+c_1^{i,j}(2)),
\]  
where $a^{i,j}$ is the entry in position $(i,j)$ of $A=G_{2}^l(2)C_2^lG_2^l(2)^\T$.

By substituting these values into $\bar{u}_{2,2}$, panel $G_1$ can therefore now compute the value $\bar{u}_{2,1}$ it needs to transmit. Proceeding in an analogous way after more  laborious but straightforward substitutions we obtain the score $\bar{U}$ for this policy, namely the polynomial:

\begin{equation*}
\bar{U}(d)=\sum_{t\in[2]}\sum_{l\in[r]}\bar{\gamma}_1(t)\tau_1^l(t)+\bar{u}_2^l(t)(\bar{\gamma}_2^l(t)+\bar{u}_3^l(t)(\bar{\gamma}_3^l(t)+\tau_{34}^l\bar{\gamma}_4^l(t)))+\tau_{mix}^l(t),
\end{equation*}
where $\bar{\gamma}_i^l(t)=k_{til}\gamma_i^l(t)$, $\tau_1^l(t)=a_1^l(t)^2+c_1^l(t)$, $\bar{u}_3^l(1)=c_{23}^l+(a_{23}^l)^2$, $\bar{u}_3^l(2)=\bar{u}_3^l(1)+r_{23}^l(2)$ and 
\begin{eqnarray*}
\tau_{mix}^l(2) &=&b_3^l(2)\bar{\gamma}_3^l(2)+\bar{\gamma}_4^l(2)\tau_4^l+2a_{04}^la_{34}^la_{23}^l\bm{a}_1(2)^\T G_2^l(2)\bm{a}_2^l,\\
\tau_{mix}^l(1) &=&b_3^l(1)\bar{\gamma}_3^l(1)+\bar{\gamma}_4^l(1)\tau_4^l+2a_{04}^la_{34}^la_{23}^l\bm{a}_1(1)^\T \bm{a}_2^l,\\
\bar{u}_2^l(2)  &=&{\bm{a}_2^l}^\T G_2^l(2)^\T(C_1(2)+\bm{a}_1(2)\bm{a}_1(2)^\T)G_2^l(2)\bm{a}_2^l+b_2^l(2)+\sum_{i,j\in[r]}(A^{i,j}+w^{i,j}_{2,l})(a^i_1(2)a_1^j(2)+c_1^{i,j}(2))\\
\bar{u}_2^l(1)  &=&{\bm{a}_2^l}^\T(C_1(2)+\bm{a}_1(2)\bm{a}_1(2)^\T)\bm{a}_2^l+b_2^l(1)+\sum_{i,j\in[r]}c_{2,l}^{i,j}(a^i_1(1)a_1^j(1)+s_1^{i,j}(1)).
\end{eqnarray*} 

There are various points we would like to draw out of this example:
\begin{enumerate}
\item even for simple systems like the one above, the required formulae are quite involved and incorporate uncertainty judgements in a subtle but appropriate way, unlike their na\"{i}ve plug-in analogues;
\item because of their symbolic integrity, the formulae to score different options and their associated message-passing operations can be hard-wired into a system just like a propagation algorithm can for a BN. These operations are usually simply substitutions, products and sums of known quantities and so fast to calculate;
\item because these methods are symbolic, in particular, the same formulae but with different inputs can be used for evaluating the scores associated with different policies. So computations can be parallellised;
\item the formulae only need the delivery of a small number of values from panels, making the necessary calculations feasible and quick once the overarching structure and utility class has been agreed.
\end{enumerate}

\subsection{UK Food Security}
\label{sec:food}
In the previous section we illustrated how the algorithms needed for an IDSS can be calculated using the structural information defining the underlying process. We now finish our discussion by outlining our ongoing work to produce such an IDSS to support UK local government officers to assess the impact that various government funding reductions might have on the food poverty of citizens for whom they have a responsibility of care. 

The first part of this task was to elicit from the decision-makers the purpose of the IDSS, the range of decision that were open to users, what the attributes of interest were, and how these were currently measured.   In the context of Warwickshire County Council, the decision space needed to be able to express how and what services to provide in order to ameliorate the negative effects of government cuts. Through a preliminary decision conference the clients identified that the bad effects of food poverty potentially manifested itself through its impacts on educational attainment, health and social incohesion expressed as the potential for inciting food riots. In constructing a utility function based on these attributes, it appeared appropriate to assume value independence \citep{Keeney1993}. Let $Y_1$ denote measures of educational attainment, $Y_2$ denote measures of health, $Y_3$ denote measures of social unrest and $Y_4$ denote cost.  An appropriate family of utility functions was then found to have the form:

\begin{equation}
U(y)=\sum_{i\in[3]} k_i(1-exp({-c_i y_i}))+k_4(a+by_4),
\label{Eq:WarwickshireUtility}
\end{equation}
where $y=(y_1, y_2, y_3, y_4)$ and whose parameters $(a, b, c_1, c_2, c_3)$ were then elicited.   

From the elicitation we could see that any structural model had to support the prediction of $(Y_1, Y_2, Y_3,Y_4)$ as these evolved into the future as a response to each potential policy choice. Figure \ref{FoodDBNTikz} gives a schematic of both the main features of the process we capture in this framework and the constitution of panels over the system. The panels constituted for such an IDSS will often be chosen to mirror the panels that are already constituted for similar purposes, e.g. in the UK, the Office for Budget Responsibility, HM Treasury and The Confederation of British Industry all produce economic forecasts (Economy). Data sources on the different components are extensive and can be incorporated into the model by appropriate panels, although the existence of some gaps in the process makes the elicitation of some probabilistic expert judgements essential.  Demography and Socio-economic status (SES) distributions are available from the Office for National Statistics (ONS); costs of housing, energy and general cost of living (CoL) via the consumer prices index (CPI); food trade (imports and exports) and farming yields from DEFRA; supply chain disruption and overall food availability from DEFRA, the food standards agency and the food and drink federation; access to credit from the Bank of England; household disposable income from ONS and Food costs via the cost of a typical basket of food, as used in the UK consumer prices index \citep{ONS2013}. We need to build separate probabilistic models of supply chains for various diverse categories of food which will then combine together to provide forecasts of the costs over time of typical baskets of food for different categories of impoverished households. At the time of writing, we have now elicited and built prototype versions of such dynamic probabilistic models for about 40\% of these categories of food \citep{Barons2014}. Once these are in place, the effect of different scenarios and strategic policies, to be elicited, can then be examined via their composite utilities.  In reality, some of this information is very rich. For example, the insurance industry has extensive data on types and locations of all the food in transit, data which is fast-moving in time, so the datasets informing these systems can be huge. But also note that typically we need only short summary statistics for the purposes of the IDSS. The distributivity of the system we build also means that different types of processes can be audited separately using analogies of the tower rule recurrences described in the last section. These can, if necessary, be adjusted through discussions with relevant domain experts.  In this way, the system becomes a vehicle through which understanding and appreciation of the input of the diverse components of the system by the user can be realised.

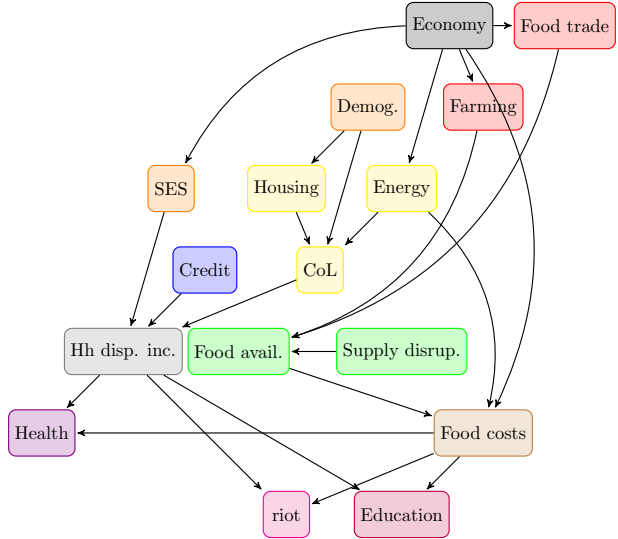
\begin{figure*}
\begin{center}
\scalebox{0.7}{
\begin{tikzpicture}[->,>=stealth',shorten >=1pt,auto,node distance=2.19cm, semithick]
\tikzstyle{every state}=[fill=white,draw,text=black, shape=rectangle, rounded corners]
\node[state](A)  [draw=black!100, fill=black!20,  text=black]              			{Economy};
\node[state](B)  [below left of = A, 	draw=orange!100, 	fill=orange!20] 						{Demog.};
\node[state](C)  [right of = A,				draw=red!100,  		fill=red!20, text=black,]{Food trade};
\node[state](D)  [below left of=B, 		draw=yellow!100, 	fill=yellow!20, text=black ]{Housing};
\node[state](E)  [right of=D, 				draw=yellow!100, 	fill=yellow!20]						{Energy};
\node[state](F)  [right of = B, 			draw=red!100, 		fill=red!20]						{Farming};
\node[state](G)  [left of= D, 				draw=orange!100, 	fill=orange!20]						{SES};
\node[state](H)  [below left of= D, 	draw=blue!100, 		fill=blue!20]						{Credit};
\node[state](I)  [right of=H, 				draw=yellow!100,	fill=yellow!20]						{CoL};
\node[state](J)  [below left of=I, 		draw=green!100,  	fill=green!20]						{Food avail.};
\node[state](K)  [below right of=I, 	draw=green!100, 	fill=green!20]						{Supply disrup.};
\node[state](L)  [below left of=H, 		draw=gray!100,  	fill=gray!20]						{Hh disp. inc.};
\node[state](M)  [below right of=K, 	draw=brown!100,  	fill=brown!20]						{Food costs};
\node[state](N)  [below left of=L, 		draw=violet!100,  fill=violet!20]						{Health};
\node[state](O)  [below left of=M, 		draw=purple!100,  fill=purple!20]						{Education};
\node[state](P)  [left of=O, 					draw=magenta!100,  fill=magenta!20]						{riot};
\path (A) edge [bend right](G)
					edge (E)
					edge [bend left](M)
					edge (C)
					edge (F)
      (B) edge (D)
          edge (I)
      (C) edge [bend left](J)
      (D) edge (I)         
      (E) edge (I)
					edge [bend left](M)
      (F) edge [bend left](J)
      (G) edge (L)
      (H) edge (L)
			(I) edge (L)
			(J)	edge (M)
      (K) edge (J)
      (L) edge (N)
					edge (O)
					edge (P)
			(M) edge (N)
					edge (O)
					edge (P);
\end{tikzpicture}
}
\end{center}
\caption{A plausible schematic of information flows for the modules of a UK food security IDSS. KEY:  Economy: UK economic forecasts; Demog.:Demography;  Farming: food production; SES: Socio-economic status; Credit: access to credit;  CoL:cost of living; Food Avail: Food availability; Supply disrup:  food supply disruption;  disp. inc.: household disposable income. \label{FoodDBNTikz} }
\end{figure*}

\subsubsection{Some illustrative detail: GCSEs and free school meals}

To illustrate how the process works, we conclude by working through in slightly more detail a simplified version of the explicit calculations used to make for the examination of inputs concerning just one variable of interest, educational attainment as measured through Impact Indicator 8 (II8), the percentage of pupils gaining 5 or more GCSE qualification, including maths and English, passed at grade C and above.   Children from low-income households are entitled to free school meals (FSM). A smaller proportion of pupils eligible for FSM achieves this benchmark compared to their peers who are not eligible.

The panel with expertise in the expected distribution of household incomes will use an appropriate probabilistic model, e.g. an MDM, to provide appropriate forecasts of numbers entitled to FSM together with measures of uncertainty, reflecting the quality of the experimental evidence on the relationship.  This enables the next panel to predict the expectation and variance of II8 under differing policy implementations. In this particular case,  since this is a terminal panel in the IDSS, its output is an attribute of the utility function, i.e. a reward. These assessments are donated to the panel calculating the effect on the future predictions of the utility function. II8 is calculated using the simple formula:
\begin{equation}
Y_1(t)=\frac{100}{n} \left(\sum_{j\in[n]} Y_{j}(t)\right),
\end{equation}
 where $Y_{j}(t)\in [0,1]$ is an indicator variable on whether or not the individual achieves the target at time $t$, the next time point. Assume an MDM holds given by: $Y_{j}(t)=A(t)-B_{j}(t) Y^{\ast}(t)+\epsilon(t),$ 
where $Y^{\ast}(t)$ is a binary indicator of eligibility for free school meals and the parameters $A(t), B(t)$ have distributions derived from the previous time point.  Then the polynomial recurrences $\mathbb{E}(Y_1(t))$ and $\mathbb{V}(Y_1(t))$ can be derived as given in Section \ref{sec:frameworks}.

Within our actual system the following calculations need to be made numerically. But as a simple illustration in order to illustrate some of functionality of the IDSS in a simple way, suppose that the reward is Gaussian. Then it is easy to check that the corresponding expected utility function given by
\begin{equation}
\bar{U}(d)=1-exp\left(-c\{\mathbb{E}(Y_1(t)\ |\ d)-\frac{1}{2}c \mathbb{V}(Y_1(t)\ |\ d)\}\right).
\label{Eq:MargUtil}
\end{equation}
This is clearly maximised when $\mathbb{E}(Y_1(t)\ |\ d)-\frac{1}{2}c \mathbb{V}(Y_1(t)\ |\ d)$ is maximised over $d \in \mathbb{D}$. So we can see here a clear trade-off between the predicted \emph{effect} of a decision against its associated uncertainty. In particular, it is apparent that we have automatically and explicitly accommodated into the scores the fact that effects must be militated by an associated uncertainty.
Note that when the reward is non-Gaussian, the expression (\ref{Eq:MargUtil}) will be a function of the cumulant generating function evaluated at different points. The education panel can, of course, simply calculate these outputs from its probabilistic model to produce the necessary summaries. 

\subsubsection{Using the IDSS to feed back information}

\begin{figure}[h!]
\centering
\includegraphics[width=0.49\linewidth]{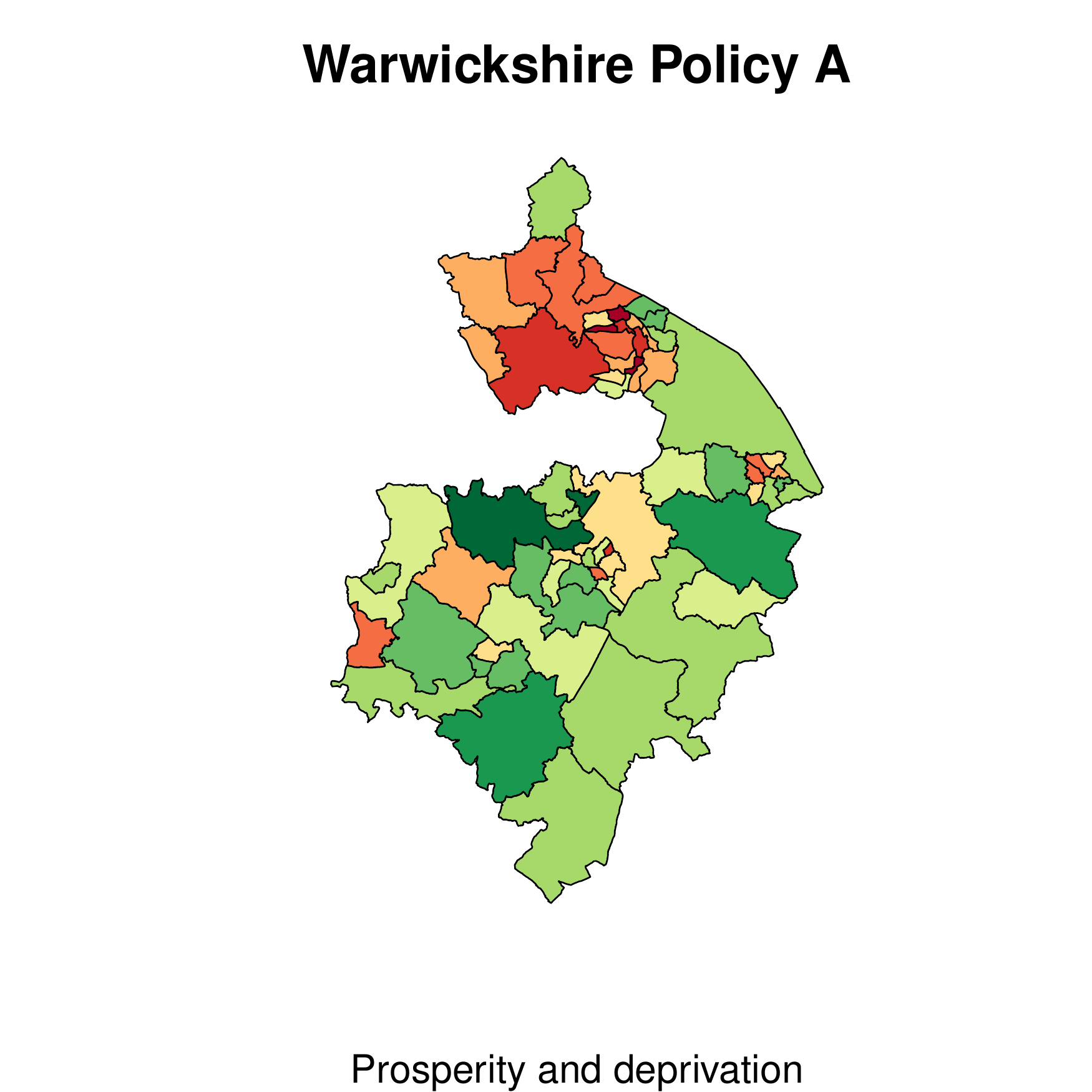}~~\includegraphics[width=0.49\linewidth]{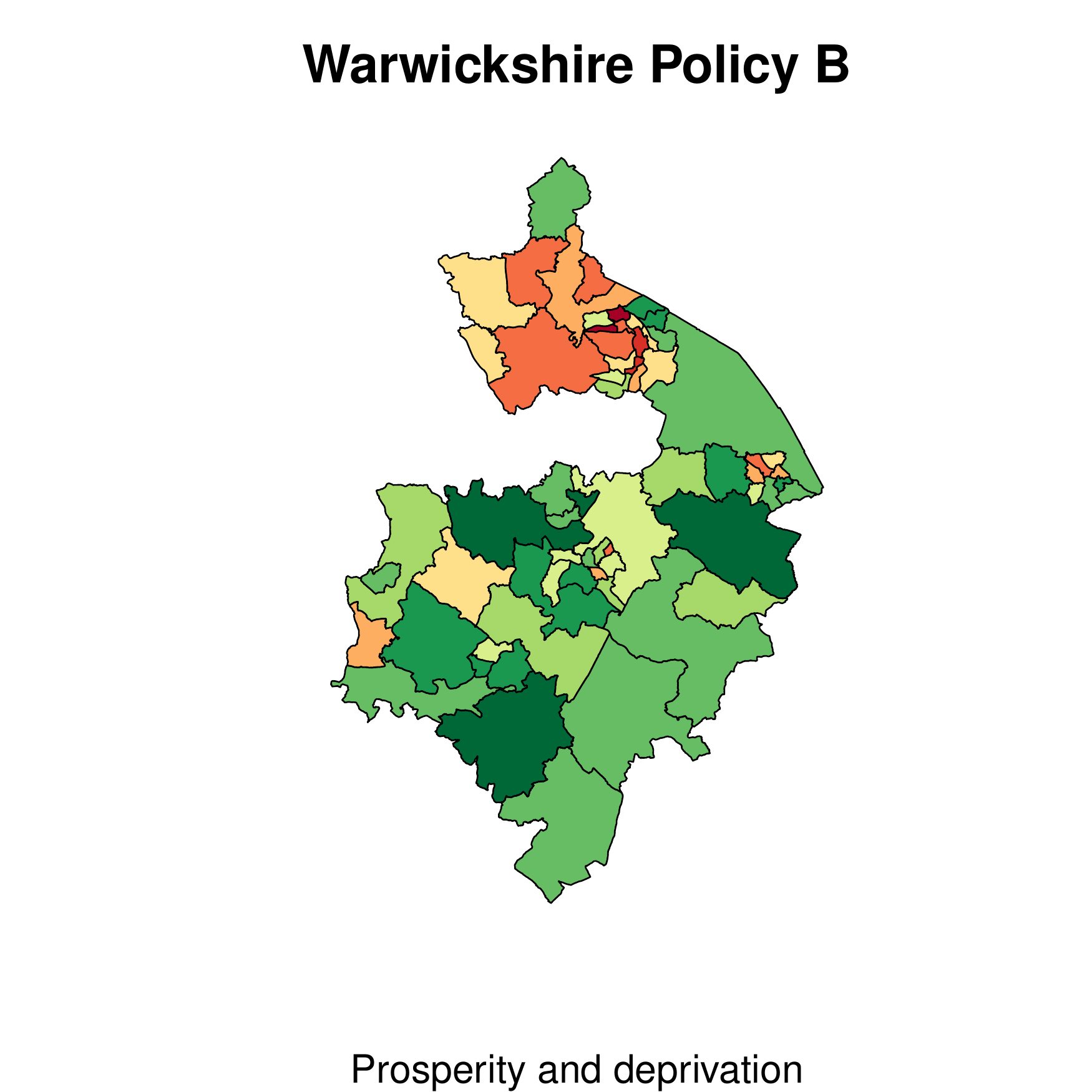}

\caption{\footnotesize UK food security: illustration of the use of regional maps for decision support. Here an indicator of prosperity shows deprived areas in red. A clustering of deprivation (left map) is a risk factor for social unrest, such as food riots.  Therefore, policies which specifically reduce and fragment large areas of poverty (right map) are to be preferred. }
\label{Fig:II8}
\end{figure}

Presently, the decision makers are presented with current and past data  concerning relevant issues displayed in graphs and maps within written reports. But these contain no annotated predictions of impact on vulnerable communities of \emph{future} events or the impact of central government changes in regulation.  Nor are evaluations given of the likely effectiveness of different implementations of changes, i.e. policy options. We are currently transferring the output of our  analyses to  the same types of graphs and maps, but now concerning \emph{predictions} of the impact of future controls.
In Figure \ref{Fig:II8} we illustrate how their familiar demographic regional map can be used to display `hot-spots' for food riot potential that might occur under certain policies. Here the decision centre can immediately see how a decision is predicted to have a dramatic deleterious effect in the more deprived areas of the region, and because of this, increase the risk of rioting. We have shown that, despite the complexities of the food security system, we have been able to establish the  CK-class for this problem, to identify how to make a sound and distributed IDSS using expert judgement based on suitable models,  data and elicited domain knowledge, to accommodate rigorously the uncertainty in the information and to provide the SB with an appropriate visualisation of the IDSS outputs to  make them usable.

\section{Discussion}

In this paper we have been able to demonstrate that a formal and feasible IDSS can be designed to enhance a decision centre's decision-making capabilities.  We noted that systems with an underlying clear causal directionality are especially amenable to this type of support.  What we usually need is that  the functionality of the support to be precisely  defined so that a suitable overarching framework can be agreed by all parties - several illustrations of these are given above - and that the IDSS has available sufficient numbers of expert judgements to inform it.

Then it will be possible to continually update the system with the best of the evidence and do this in a distributed way that ensures it remains coherent, sound and feasible. Furthermore, the calculations of scores of policy options will often depend on the expectations of a few functions donated by autonomous panels, communicated to the other panels in simple ways so that uncertainties are appropriately incorporated into the analysis.

Of course, as for any class of statistical model, the validity of the assumptions underpinning its evaluation needs to be checked against both domain knowledge and empirical evidence as this becomes available. As we stated earlier, under the conditions we derive above diagnostic checks of the various component models can be devolved to the responsible panel. But what if the structural assumptions within the CK class appear from data to be violated even though the individual components themselves perform well?

This is a challenging issue and obviously depends critically on both the domain on which the IDSS applies and the nature of the information sets. In crisis management systems like RODOS, as an incident unfolds such structural reappraisal would be very difficult to enact in real time. However, within integrating systems for policy decision support, like the food security system discussed above, such data-driven appraisal and revision may be possible. The obvious way to proceed would be to set the predicted consequences of policy choices made using the system against their predictions. A natural framework for applying these techniques would be the prequential approach \citep{Dawidpre} where model predictions of immediate consequences of enacted policies are compared with what was actually observed to happen - for example in terms of comparisons of observed and predicted marginal utility scores - a methodology we plan to fully develop in the future.

Our focus here has been on networking probabilistic models. However this
focus was determined mainly because probabilistic decision support tools are
the most widely used of those that accommodate uncertainty in some way. We
accept there are other modelling tools - for example those based on linear
Bayes or belief function inferential schemes - which also have their own
formal structures and using similar technologies can also ensure coherent
and speedy integrating decision support. 

We have noted above that many of the computational forecasting formulae can
be expressed algebraically where the systems of equations can be deduced
from the elicited overarching system. So, in particular, computer algebra can be used to code up  these equations which
can then be examined for their algebraic properties. As well as guiding the design of the network architecture of an IDSS, issues like the
robustness of the systems to variation in inputs can be formally and
systematically analysed using these tools. We have already made some
tentative steps in this development \citep[see][]{Leonelli2015b}.

Of course, in some situations the likelihood separation we need to
ensure the enduring formal distributivity of a formal IDSS will break down.
Then we will need to fall back on approximate inferential methods that
preserve distributivity. However preliminary results suggest that sensible
approximate methods - whose form is similar to variational Bayes methods -
can still work effectively provided that statistical diagnostics are
employed to check that the family of models described by the overarching
posited structure remains plausible in the light of the information
available. These results will be reported in due course.

We believe that, with the increasing demand for complex evidence based
models the demand for integrating decision support systems can only
increase. If such systems do not encode uncertainty intelligently then as
statisticians we know that these will be liable to grossly mislead. But as
we demonstrate here it is straightforward to appropriately encode
uncertainties into such a framework and in this way to properly appraise the available options.

\section*{Acknowledgements}
JQS \& MJB were funded by EPSRC grant EP/K039628/1, `Coherent inference over a network of probabilistic systems for decision support with applications to food security' and ML was funded by EPSRC grant EP/K007580/1, `Management of Nuclear Risk Issues: Environmental, Financial and Safety', within the NRFES project.

\bibliographystyle{plainnat} 

\bibliography{IDSS} 

\section{Appendix}
\subsection{Proof of Theorem \ref{theo:gold}}
\label{gold}
Fix a $d\in \mathcal{D}$ and for simplicity suppress this dependence and the time index $t$. We need show that under the conditions of the theorem at any time $t$ and under any policy the SB will continue to hold panel independent beliefs, i.e. that, for $i\in[m]$, 
\begin{equation}
\bm{\theta }_{i}\independent \bm{\theta }_{i^{-}}\;|\;I_{+},
\label{indep components}
\end{equation}
and that when assessing $\bm{\theta }_{i}$, she will only use the information $G_{i}$ would use if acting autonomously in assessing the information she needs to deliver, i.e. that
\begin{equation}
\bm{\theta }_{i}\independent I_{+}\;|\;I_{0},I_{ii}.  \label{indep updating}
\end{equation}
This is then sufficient for soundness and distributivity. Because even if all panellists could share each other's information then, given all panel beliefs, they would come to the same assessment about the joint distribution of the relevant parameters: that all panel subvectors are mutually independent of each other and that these margins will simply be the margins of the associated panel were they making decisions autonomously. 

The proof simply uses the semi-graphoid axioms of conditional independence as stated in \citet{Smith1989}:\label{Def:SemiGraphoid}\\
\emph{Symmetry:} For any three vectors of measurements, $X,Y,Z:$ 
\[X \independent Y\ | \ Z \Leftrightarrow Y\independent X \ |\ Z
\]
\emph{Perfect composition:} For any four disjoint vectors of measurements, $W,X,Y,Z:$  
\[
X \independent (Y,Z)\ |\ W \Leftrightarrow  X \independent Y \ | \ (W,Z)\  \&\  X \independent Z\ | \ W
\]

 We start by proving the panel independence condition in (\ref{indep components}). Note that from common separability in equation (\ref{commonly sep}) it follows that
\begin{align}
\bm{\theta}_i &\independent \bm{\theta}_{i^{-}} \;|\; I_{0}.\nonumber
\shortintertext{which combined with the separately informed condition in equation (\ref{sep inform}) through perfect decomposition and using the symmetric property of semi-graphoids axioms gives}
\bm{\theta}_{i^{-}} &\independent I_{ii},\bm{\theta}_i\;|\; I_{0}.\nonumber 
\shortintertext{ Using again perfect decomposition and symmetry, it follows that} 
\bm{\theta}_i&\independent \bm{\theta}_{i^{-}}\;|\; I_{ii},I_{0}. \label{proof:sound1}\\
\shortintertext{Now the cutting condition in equation (\ref{cutting}) together with equation (\ref{proof:sound1}) implies by perfect decomposition that}
\bm{\theta}_i&\independent I_*,\bm{\theta}_{i^{-}}\;|\; I_{0},I_{ii}.\label{proof:sd}\\
\shortintertext{Then again by perfect decomposition we have that}
\bm{\theta}_i&\independent \bm{\theta}_{i^{-}}\;|\; I_{0},I_{ii}, I_*.\nonumber\\
\shortintertext{Since $I_{ii}$ is a function of $I_*$ the above expression is equivalent to} 
\bm{\theta}_i&\independent \bm{\theta}_{i^{-}}\;|\; I_{0}, I_*. \label{proof:sound2} \\
\shortintertext{The delegatable condition in equation (\ref{delegatable}) can be written as}
I_+&\independent \bm{\theta}_i,\bm{\theta}_{i^{-}}\;|\; I_{0},I_{*},\nonumber
\shortintertext{so, using perfect decomposition, it follows that}
I_+&\independent \bm{\theta}_{i}\;|\; I_{0},I_*,\bm{\theta}_{i^{-}}. \label{proof:sound3}
\shortintertext{Combining via perfect decomposition equations (\ref{proof:sound2}) and (\ref{proof:sound3}) and using the symmetry property we have that}
\bm{\theta}_{i}&\independent I_+,\bm{\theta}_{i^{-}}\;|\; I_{0},I_*. \nonumber\\
\shortintertext{So using again perfect decomposition it follows that}
\bm{\theta}_{i}&\independent \bm{\theta}_{i^{-}}\;|\; I_{0},I_*,I_+,\nonumber
\end{align}
which, since $I_*$ and $I_{0}$ are functions of $I_+$, can be written as equation (\ref{indep components}). This shows that panel independence directly follows from the four conditions of Definition \ref{def:cond}.

To show that equation (\ref{indep updating}) holds, note that another implication of delegatability in equation (\ref{delegatable}) by perfect decomposition is that
 \begin{align}
\bm{\theta}_i&\independent I_{+}\;|\; I_{0},I_{ii},I_*, \label{proof:sound41}\\
\shortintertext{where again we used the fact that $I_{ii}$ is a function of $I_*$. Now noting that by perfect decomposition equation (\ref{proof:sd}) implies that}
\bm{\theta}_i&\independent I_* \;|\; I_{0},I_{ii}, \label{proof:sound5}\\
\shortintertext{it follows from equations (\ref{proof:sound41}) and (\ref{proof:sound5}) by perfect decomposition that}
\bm{\theta}_i & \independent I_+,I_*\;|\; I_{0},I_{ii},\nonumber
\end{align}
which, since $I_*$ is a function of $I_+$ is equivalent to equation (\ref{indep updating}).

\subsection{Proof of Theorem \ref{theo:seplik}}
\label{proof:seplik}
Fix a policy $d\in\mathcal{D}$ and suppress for ease of notation this dependence. Under the initial hypotheses, by Theorem \ref{theo:gold}
\begin{equation*}
\independent_{i\in[m]}\bm{\theta }_{i}\;|\;I_{0}^{0},I_{\ast }^{0},
\end{equation*}
implying that the prior joint density can be written in a product form 
\begin{align*}
\pi (\bm{\theta })&=\prod_{i\in[m]}\pi _{i}(\bm{\theta }_{i}).\\
\shortintertext{It follows that under this admissibility protocol} 
\pi (\bm{\theta }\;|\;\bm{x}^{t})&=\prod_{i\in[m]}\pi_{i}(\bm{\theta }_{i},\bm{t}_{i}(\bm{x}^{t})),\\
\shortintertext{where from the form of the likelihood above} 
\pi _{i}(\bm{\theta }_{i},\bm{t}_{i}(\bm{x}^{t}))&\propto l_{i}(\bm{\theta }_{i}\;|\;\bm{t}_{i}(\bm{x}^{t}))\pi _{i}(\bm{\theta }_{i}).
\end{align*}
By hypothesis $\pi _{i}(\bm{\theta }_{i},\bm{t}_{i}(\bm{x}^{t}))$ will be delivered by $G_{i}$ and adopted by the SB. So the IDSS is sound. In particular we can deduce, through the definition of panel separability and the results of Theorem \ref{theo:gold}, that 
\begin{equation*}
\independent_{i\in[m]}\bm{\theta }_{i}\;|\;I_{0}^{t},I_{+}^{t}\;\;\Longleftrightarrow\; \;\independent_{i\in[m]}\bm{\theta }_{i}\;|\;I_{0}^{t},I_{\ast }^{t}.
\end{equation*}
So we have separability a posteriori. Note that in the notation above 
\begin{equation*}
I_{+}^{t} =\left\{I_{0}^{t},I_{\ast}^{t}\right\} =\left\{ I_{0}^{0},I_{\ast }^{0},\bm{x}^{t}\right\},
\end{equation*}
and, for $i\in[m]$,
\begin{equation*}
\left\{ I_{0}^{t},I_{ii}^{t}\right\} =\left\{ I_{0}^{t},I_{ii}^{0},\bm{t}_{i}(\bm{x}^{t})\right\}.
\end{equation*}
Since by definition $\bm{t}_{i}(\bm{x}^{t})$ is known to $G_{i}$, $i\in[m]$,  the system is also delegatable. Finally note that if 
\begin{equation*}
l(\bm{\theta }\;|\;\bm{x}^{t})\neq\prod_{i\in[m]}l_{i}(\bm{\theta }_{i}\;|\;\bm{t}_{i}(\bm{x}^{t}))
\end{equation*}
on a set $A$ of non zero prior measure, then the conditional density $\pi_{A}(\bm{\theta }\;|\;\bm{x}^{t})$ on $A$ will have the property that for all $\bm{\theta} \in A$ 
\begin{equation*}
\pi _{A}(\bm{\theta }\;|\;\bm{x}^{t})\neq\prod_{i\in[m]}\pi _{A,i}(\bm{\theta }_{i},\bm{t}_{i}(\bm{x}^{t})),
\end{equation*}
where $\pi_{A,i}$ denotes the density delivered by panel $G_i$ for the parameters it oversees in $A$. This means that panel parameters are a posteriori dependent and so in particular the density determined by the margins is not sound.

\end{document}